\newcommand{\vars}{\operatorname{vars}}

\documentclass[a4paper,UKenglish,cleveref, autoref, thm-restate]{lipics-v2021}

\title{Proofdoors and Efficiency of CDCL Solvers} 

\author{Sunidhi {Singh}}{Georgia Institute of Technology, USA  }{sunidhisingh@gatech.edu}{}{}

\author{Vincent Liew}{Independent Researcher}{vliew@vliew.com}{}{}

\author{Marc Vinyals}{University of Auckland, New Zealand}{marc.vinyals@auckland.ac.nz}{}{}

\author{Vijay Ganesh}{Georgia Institute of Technology, USA}{vganesh45@gatech.edu}{}{}

\authorrunning{S. Singh, V. Liew, M. Vinyals, and V. Ganesh} 

\Copyright{Sunidhi Singh, Vincent Liew, Marc Vinyals, and Vijay Ganesh} 



\category{} 

\relatedversion{} 

\acknowledgements{We would like to thank Antonina Kolokolova, Noah Fleming, Robert Robere, Ryan Williams, Moshe Vardi, and Ken McMillan }

\usepackage{tikz}
\usetikzlibrary{positioning,arrows.meta,fit,calc,shapes.callouts}
\usepackage{adjustbox} 
\usepackage{longtable}
\usepackage{array}
\usepackage{amsmath}
\newcommand{\poly}{\mathrm{poly}}
\usepackage{refcount}

\newcommand{\MultGraph}{\textsc{Mult-Graph}}

\newcommand{\AddGraph}{\textsc{Add-Graph}}

\newcommand{\Vbefore}{V_{\mathrm{before}}}
\newcommand{\Vafter}{V_{\mathrm{after}}}
\newcommand{\Vshared}{V_{\mathrm{shared}}}

\newcommand{\EQ}{\mathrm{EQ}}
\newcommand{\Res}{\mathsf{Res}}

\EventEditors{John Q. Open and Joan R. Access}
\EventNoEds{2}
\EventLongTitle{42nd Conference on Very Important Topics (CVIT 2016)}
\EventShortTitle{CVIT 2016}
\EventAcronym{CVIT}
\EventYear{2016}
\EventDate{December 24--27, 2016}
\EventLocation{Little Whinging, United Kingdom}
\EventLogo{}
\SeriesVolume{42}
\ArticleNo{23}

\begin{document}

\maketitle

\begin{abstract}

We propose a new parameter called \emph{proofdoor} in an attempt to explain the efficiency of CDCL SAT solvers over a certain class of formulas derived from circuit (esp., arithmetic) verification applications. Informally, given an unsatisfiable CNF formula $F$ over $n$ variables, a {\it proofdoor decomposition} consists of a chunking of the clauses into $A_1,\ldots,A_k$ together with a {\it sequence of interpolants connecting} these chunks.
Intuitively, a proofdoor captures the idea that an unsatisfiable formula can be refuted by reasoning chunk by chunk, while maintaining only a summary of the information (i.e., interpolants) gained so far for subsequent reasoning steps.

We prove several theorems in support of the proposition that proofdoors can explain the efficiency of CDCL solvers for some class of circuit verification problems. 
First, we show that formulas with \emph{small proofdoors} (i.e., where each interpolant is $O(n)$ sized, each chunk $A_i$ has small pathwidth, and each interpolant clause has at most  $O(\log n)$ backward dependency on the previous interpolant), have short resolution (Res) proofs and a certain configuration of CDCL solvers can compute such proofs in time polynomial in $n$. Second, we show that commutativity (miter) formulas over floating-point addition have small proofdoors and hence short Res proofs, even though they have large pathwidth. Third, we identify limits of the proofdoor framework: we show that a poor decomposition of arithmetic miter instances can force exponentially large interpolants, and hence our framework derives exponentially large Res refutations from such a decomposition, even when a different decomposition
(i.e., a small proofdoor) yields short proofs. As a byproduct, these interpolant lower bounds imply new lower bounds for the partially ordered resolution proof system.

\end{abstract}

\section{Introduction}

The Boolean Satisfiability Problem (SAT) is the canonical NP-complete problem ~\cite{10.1145/800157.805047}, and it is conjectured to require exponential time in the worst case under the Exponential Time Hypothesis (ETH)~\cite{DBLP:journals/jcss/ImpagliazzoP01}. Modern SAT solvers are predominantly based on Conflict-Driven Clause Learning (CDCL) and it has been shown that they are polynomially equivalent to the resolution (Res) proof system under non-deterministic branching and restarts~\cite{pipatsrisawat2009power}. Further, it was recently shown that the problem of automating Resolution, i.e., finding a short refutation even when one exists, is NP-hard~\cite{DBLP:journals/jacm/AtseriasM20}. From this perspective, complexity theory predicts strong limitations on the scalability of CDCL solvers.

Despite these barriers, modern CDCL solvers routinely solve large real-world instances arising from hardware verification, software analysis, and planning, often involving millions of variables and clauses ~\cite{clarke2001bounded,kautz1992planning,xie2005saturn}. This  gap between worst-case hardness and empirical performance motivates the central question at the intersection of theory and practice of SAT solving, namely, {\it what properties of real-world SAT instances make them tractable for CDCL?}

A natural explanation for this apparent contradiction is that NP-completeness captures only worst-case behavior. By contrast, the efficiency of CDCL is witnessed on large, structured instances that arise in practice ~\cite{biere2009handbook}. However, despite decades of effort, we still lack a rigorous mathematical characterization of which real-world instances are tractable for CDCL and why they admit proofs that CDCL can reliably discover.

In this work, we attempt a proof-theoretic explanation for the efficiency of CDCL
based on the observation that solvers often reason through formulas in a sequential, local manner. Rather than solving the formula globally, CDCL  progressively processes parts of the instance while maintaining a small
summary of the information learned so far. We introduce the notion of
\emph{proofdoors}, which capture such reasoning via a sequential decomposition
of a formula together with small interpolants summarizing the information
learned so far. We show that formulas admitting \emph{small proofdoors} -- a
special case of proofdoors -- have short resolution refutations, and that under suitable (non-deterministic) configurations CDCL can efficiently refute such formulas. We
also analyze the limitations of such decompositions.

\paragraph*{Existing Explanations}

Early attempts to explain SAT-hardness focused on statistical properties of random $k$-SAT formulas, such as the clause-variable ratio. Experiments showed that as this ratio increases, the formulas undergo a phase transition, i.e., below a certain ratio they are typically satisfiable and above the ratio they are typically unsatisfiable, and solver runtimes peak near the threshold~\cite{mitchell1992hard,cheeseman1991really}. This suggested that hardness might arise from formulas lying near this region. However, these phenomena apply primarily to random instances and fail to explain the difficulty of industrial formulas, whose hardness is largely uncorrelated with density-based measures.

Another line of work studies graph-theoretic width parameters such as treewidth, pathwidth, and cutwidth of the primal or incidence graph of a CNF formula. For formulas of bounded width, SAT can be solved in polynomial time via the Davis–Putnam algorithm~\cite{alekhnovich2002satisfiability,davis1958feasible}. Huang and Darwiche augment DPLL with a caching scheme to construct 
OBDDs for CNF formulas, yielding complexity exponential only in the 
cutwidth or pathwidth of the variable ordering~\cite{huang2004dpll}. 
Prasad et al. also show that caching-based SAT algorithms with static 
variable orderings run in time $O(n \cdot 2^{2k_{\mathrm{fo}} \cdot w})$ 
where $w$ is the cutwidth and $k_{\mathrm{fo}}$ the maximum 
fanout~\cite{prasad2001combinational}. Wang et al. use a modified learning scheme on a DPLL style solver, yielding a time complexity exponential only in the cutwidth $w$ ~\cite{wang2001cutwidth}. However, real-world formulas typically have large  width yet remain easy for CDCL, for example, BMC formulas \cite{biere2009bounded}, indicating that  bounded width alone does not explain solver performance.

Yet another set of theoretical concepts is backdoors and backbones, which capture the idea that a small set of variables controls the complexity of a formula. Backdoors identify variables whose assignment reduces the formula to a tractable class, while backbones consist of variables that take the same value in all satisfying assignments~\cite{kilby2005backbones}. While these notions admit strong theoretical guarantees, empirical studies suggest that industrial instances often have large backdoors, and backdoor size correlates weakly with CDCL runtime. In contrast, backbone variables appear more frequently in practice, but their presence alone does not explain how CDCL efficiently discovers the relevant assignments \cite{zulkoski2018learning}.

One empirical explanation is based on community structure, where formulas decompose into loosely connected clusters of variables and clauses.  In their seminal work, Ans'otegui et al.~\cite{ansotegui_community_2019} established that most industrial instances exhibit high modularity, and proposed a modularity-based preprocessor that forces the solver to learn clauses connecting pairs of communities, yielding empirical improvements especially on satisfiable instances. Several studies report correlations between solver performance and measures of modularity or hierarchical structure, suggesting that CDCL benefits from local reasoning within communities~\cite{newsham2014impact}. However, there exist strong counterexamples showing that community structure cannot be used as the foundation of a theory~\cite{mull2016hardness}.

Another suggested parameter, mergeability, measures the extent to which resolved clauses can be combined during proof construction. Merge-based proof systems admit theoretical analysis and yield insights into certain solver behaviors \cite{vinyals_et_al:LIPIcs.SAT.2023.27}. Yet it is unclear whether industrial instances consistently exhibit high mergeability, and empirical evidence connecting merge counts to solver efficiency is limited \cite{zulkoski2018effect}.

From a proof-theoretic perspective, CDCL solvers generate resolution proofs, and lower bounds on resolution proof size translate directly to lower bounds on solver runtime. Such results successfully explain solver failures on well-known hard families, including Tseitin formulas and the pigeonhole principle \cite{haken1985intractability,tseitin1983complexity}, but offer little insight into why many industrial formulas admit short resolution proofs that CDCL can find efficiently in practice. 

To summarise, despite decades of work, no existing parameter simultaneously captures the structure of real-world SAT instances and provides rigorous guarantees on proof size or proof search, leaving a persistent gap between empirical observations and theoretical understanding.


\begin{enumerate}
    \item {\bf Small Proofdoors and Upper Bounds on Resolution Proof Size:}  Our main structural result shows that if an unsatisfiable formula has small proofdoors, it admits a polynomial-size resolution refutation (Theorem~\ref{thm:proofdoor-main}, Corollary~\ref{thm:smallpfd}, Proofdoor Theorem, Small Proofdoor Corollary).
    

    \item {\bf FP Small Proofdoors:} We show SAT encodings of commutativity of floating point addition satisfy the conditions of our main theorem and thus admit short proofs (Theorem~\ref{thm:smallpfdfloat}, FP Small Proofdoor Theorem).

    \item {\bf Proofdoors and Lower Bounds:} 
    We demonstrate that the choice of decomposition is crucial by proving exponential interpolant size lower bounds corresponding to
    a natural class of decompositions of arithmetic miter instances
    (Theorem~\ref{thm:miter_lower_bound}, Miter Lower Bound). This forces the refutations derived by Theorem~\ref{thm:proofdoor-main} to be exponentially large, even when a different decomposition
    (i.e., a small proofdoor) yields short proofs (Theorem~\ref{thm:proofdoor_lb}). As a byproduct, we show
    that these interpolant size lower bounds imply lower bounds for the partially
    ordered resolution proof system (Lemma~\ref{lem:interpolants}),
    generalizing a lower bound of Janota~\cite{janota2016exponential}.

    \item {\bf Undecidability:}  Finally, we show that deciding whether a general formula family admits polynomial-size proofs is undecidable (Theorem~\ref{thm:fpsp}, FPSP Undecidability Theorem), highlighting fundamental limits on any complete characterization of the efficiency of CDCL solvers.
\end{enumerate}

\paragraph*{Related Work}

A primary inspiration for our work came from SAT encodings of arithmetic identities, such as the commutativity of integer multiplication (i.e., $xy \neq yx$), as studied in \cite{beame2019toward}. These instances surprisingly admit short resolution proofs despite being conjectured otherwise. The key proof idea is to decompose the formula into structured chunks corresponding to the circuit components and derive the equality of the output bits incrementally. Each chunk involves only a small pathwidth subcircuit, and the final proof of unsatisfiability can be obtained by combining the local proofs. This can also be understood as computing interpolants between neighboring chunks. Our Theorem  \ref{thm:proofdoor-main} generalizes this insight that formulas can be decomposed into bounded pathwidth fragments where each interpolant clause depends on only $O(\log n)$ previous interpolant clauses.  Miter encodings for commutativity of integer multiplication fall within this framework (see Appendix~\ref{app:miterasthm1}).

The classical Davis Putnam algorithm  \cite{davis1958feasible}, that is the foundation for all CDCL SAT solvers, operates by eliminating variables in a fixed order through Resolution, effectively computing intermediate resolvents that are implied by a part of the formula. The elimination process can be viewed as computing interpolants; when a variable is eliminated, the remaining clauses encode all the consequences of the formula only over the surviving variables. While Davis-Putnam is not polynomial-time in general, makes no attempt to compute small proofdoors, and may incur exponential blowup, upon reflection, it serves as the first example of proof generation via proofdoors. 

Another relevant work is the well known paper by Ken McMillan~\cite{mcmillan2003interpolation} on interpolation model checking, which introduced the explicit use of interpolants extracted from resolution proofs to refine over-approximations in SAT-based bounded model checking (BMC). In that setting, interpolants summarize the information obtained from one unrolling of the transition system and are used to strengthen the constraints in the next unrolling. However, their work didn't attempt to explain the power of CDCL SAT solvers. Further, their system explicitly computes interpolants to speedup the underlying solver. By contrast, we provide mathematical proofs that certain class of verification instances admit small proofdoors and hence short resolution proofs. Further, an idealized version of CDCL can leverage them to find short proofs.

\section{Preliminaries}

We denote Boolean variables by lowercase letters (e.g., $x$). Bit-vectors are written in boldface (e.g., $\mathbf{x}$), with subscripts denoting individual bits (e.g., $x_i$). Other sets of variables are denoted by uppercase letters (e.g., $S$). A literal is a Boolean variable or its negation. A clause is a disjunction (or set) of literals, and a CNF formula is a conjunction (or set) of clauses. For a clause $C$ or formula $F$, we write $C\upharpoonright_{\alpha}$ or $F\upharpoonright_{\alpha}$
for their restrictions under $\alpha$.

A \emph{resolution refutation} of a CNF formula $F$  is a sequence of clauses $\pi=(D_1,\dots,D_L)$ with $D_L=\bot$ such that each $D_i$ is either a clause of $F$, or is obtained by resolving two earlier clauses on some variable $x$, i.e.,
$D_i=\Res_x(D_j,D_k)$ for some $j,k<i$, where $\Res_x(B\vee x,\; C\vee \neg x):=B\vee C$.

Resolution can be visualized as either a sequential list of clauses or a directed acyclic graph (DAG), where each vertex represents a clause and edges represent the derivation dependencies. A resolution refutation is \emph{ordered} if along any leaf-to-root path in its DAG, the variables are resolved in an order consistent with a fixed total ordering of the variables. 

The notion of ``proofdoors'' for SAT solving is closely connected to \emph{partially ordered resolution}~\cite{janota2016exponential}, a relaxation of ordered resolution in which the variable order is a partial order rather than a total order. 

\begin{definition}
    Let $\prec$ be a partial order on the variables of a CNF formula $\Phi$. A resolution proof $\pi$ is { \em $\prec$-ordered} if, on every leaf-to-root path $P$ in $\pi$, the sequence of resolved variables respects the partial ordering $\prec$: whenever $x\prec y$ and both $x$ and $y$ are resolved on along $P$, the path $P$ resolves on $x$ before $y$.
\end{definition}

\begin{definition}
    Let $X, Y$ be disjoint sets of Boolean variables. Denote by $X \prec Y$ the partial ordering $\{x \prec y | x \in X, y \in Y\}$.
\end{definition}


\begin{definition}[Interpolant \cite{craig1957three}]
Let $\Phi(X,Y,Z) = A(X,Z) \wedge B(Y,Z)$ be an unsatisfiable CNF formula. An \emph{interpolant} for $(A,B)$ is a Boolean function $I(Z)$ over the shared variables $Z$ such that $A(X,Z)\models I(Z)$ and $I(Z)\wedge B(Y,Z)$ is unsatisfiable.
\end{definition}
Depending on context, ``interpolant'' may refer to either the Boolean function $I(Z)$ itself, or to a CNF formula that computes $I$. This notion captures the idea of separating a formula into two parts and summarizing the contribution of one part to the overall contradiction. In this sense, interpolants formalize the idea of reasoning across a cut
using only the information that must be communicated between the two sides.


\noindent To reason about the structural complexity of formulas and proofs, we will make use of the notion of pathwidth. 

\begin{definition}[Pathwidth \cite{robertson1983graphminors}]
A \emph{path decomposition} of an undirected graph $G=(V,E)$ is a sequence of
vertex sets $(B_1,\dots,B_m)$ , called \emph{bags},  such that:

\begin{enumerate}
\item $\bigcup_i B_i = V$,
\item for every edge $\{u,v\}\in E$, some $B_i$ contains both $u$ and $v$,
\item for every $v\in V$, the indices $\{ i \mid v\in B_i \}$ form a contiguous interval.
\end{enumerate}

\noindent The \emph{pathwidth} of $G$ is $\min_{\mathcal{P}}  \max_i (|B_i|-1)$ over all path decompositions $\mathcal{P}$.
\end{definition}

\begin{definition}[Clause Variable Incidence Graph]
Let $F$ be a CNF formula with variable set $X$ and clause set $\mathcal{C}$.
The \emph{clause variable incidence graph} of $F$ is the bipartite graph
$G_F = (X \cup \mathcal{C},\, E)$ where $\{x,C\} \in E$ if and only if
the variable $x$ occurs (positively or negatively) in the clause $C$.
\end{definition}

The pathwidth of a CNF formula refers to the pathwidth of its clause variable incidence graph. 



\section{Proofdoors}
We now introduce the central notion of this work by generalizing the idea of a single interpolant between two partitions of an unsatisfiable formula into a sequence of interpolants, allowing us to reason through a formula chunk by chunk. Intuitively, a proofdoor captures the idea that an unsatisfiable formula can be refuted by reasoning sequentially through its components or chunks, while maintaining only a summary of the information gained so far.

\begin{definition}[\bf{Proofdoor Decomposition}]\label{def:proofdoordecomp}

A \emph{proofdoor decomposition} of an unsatisfiable CNF formula $F$ is an expression 
$$F \;=\; A_1(L_1,S_1)\ \wedge\ A_2(L_2,S_2)\ \wedge\ \cdots\ \wedge\ A_{k-1}(L_{k-1},S_{k-1})\ \wedge\ A_k(L_k),
$$
together with a sequence of interpolants 
$I_1(S_1), I_2(S_2), \dots, I_{k-1}(S_{k-1}),$ represented as CNFs, where each $A_i$ is a CNF formula over the indicated variable sets, each $L_i$ is the set of variables whose last appearance is in $A_i$, and each
$S_i := \vars(A_1 \wedge \cdots \wedge A_i) \cap
              \vars(A_{i+1} \wedge \cdots \wedge A_k)$
is the set of variables shared across the $i$-th cut. The interpolants satisfy that $I_1$ interpolates from $A_1$ to $A_{2} \land \dots \land A_k$, and each subsequent $I_j$ interpolates from $I_{j-1} \land A_j$ to $A_{j+1} \land \dots \land A_k$.

Define $R_i$ to be the set of variables whose first appearance is in a chunk $A_j$ with $j > i$. A proofdoor decomposition induces the set of \emph{cutting partial orders}
$\bigl\{\, L_i \prec\ R_i \bigr\}$ for $i = 1, \ldots, k-1$.
\end{definition}

 This imposes no restrictions on the size of the interpolants or number of interpolants or on the structure of the decomposition.

\begin{definition}[\bf{Proofdoors}]\label{def:proofdoors}
Let $F$ be an unsatisfiable CNF formula over $n$ variables, and let $c,w,s,k \ge 1$ be integers.
We say that $F$ \emph{admits a proofdoor with parameters $(c,w,s,k)$}
if there exists a proofdoor decomposition into  $k$ chunks
$F = A_1 \wedge \cdots \wedge A_k$
with interpolants $(I_1,\ldots,I_{k-1})$ satisfying:

\begin{enumerate}

\item $\forall j \in \{1,\ldots,k-1\}$, the interpolant $I_j$ contains at most $c$ clauses.

\item $\forall j \in \{2,\ldots,k-1\}$ and every clause $C \in I_j$,
there exists a set $S(C) \subseteq I_{j-1}$ of size at most $s$ such that
$A_j \wedge S(C) \models C$.  We call $S(C)$ the \emph{backward dependency} 
of $C$ on $I_{j-1}$.

\item The interpolant $I_{k-1}$ contains at most $s$ clauses.

\item $\forall j \in \{1,\ldots,k\}$, the chunk $A_j$ has clause-variable incidence pathwidth at most $w$.

\end{enumerate}
\end{definition}

\begin{remark}
Every unsatisfiable CNF formula $F$ over $n$ variables and $m$ clauses 
trivially admits a proofdoor with parameters $(c,w,s,k) = (0,\, m+n,\, 0,\, 1)$: 
take $k=1$ so $F = A_1$ with no interpolants, and the conditions on $c$ and $s$ 
are vacuously satisfied. The pathwidth of $A_1$ is at most $m+n$. 
\end{remark}

The conditions above are independent. If the interpolants are unrestricted in size, a proofdoor may exist trivially by doing DP elimination over all non-shared variables from the prefix. Even if each chunk has a small pathwidth, 
small interpolants need not exist without additional structural
constraints. In particular, arbitrary partitions of $F$ do not
guarantee small proofdoors. For example, partitioning the formula into
single clauses does not prevent the interpolants from becoming large,
and placing most clauses on one side of a cut may induce large pathwidth.
Thus, both bounded interpolant size and bounded pathwidth are required.  Note that this definition allows the interpolant clauses to repeat among the interpolants. We consider the following decision problem. 

\begin{theorem}[\bf{Proofdoor computation is NP-hard}]\label{np:hard}
Determining whether an unsatisfiable CNF formula $F$,
together with a partition $F = A_1 \wedge \dots \wedge A_k$
and integers $c,w,s,k$, admits a proofdoor
with respect to $c,w,s$, and $k$
is NP-hard.
\end{theorem}

\begin{proof}[Proof Sketch]
Consider the special case $k = 1$, so $F = A_1$.
Then the question reduces to whether $F$ has pathwidth at most $w$.
Since deciding whether a bipartite graph has pathwidth at most $w$
is NP-hard \cite{bodlaender1991approximating}, the result follows.
\end{proof}


\noindent We now give a structural condition under which a particular configuration of the CDCL SAT solver can efficiently refute a formula.

\begin{definition}[\bf{Small Proofdoors}]\label{def:smallpfd}
An unsatisfiable CNF formula $F$ with $n$ variables
admits \emph{small proofdoors} if it has a proofdoor
with parameters $(c,w,s,k)$ satisfying $c = O(n),  w = O(\log n), s = O(\log n), k = O(n)$.
\end{definition}

\section{Proofdoor Theorem}

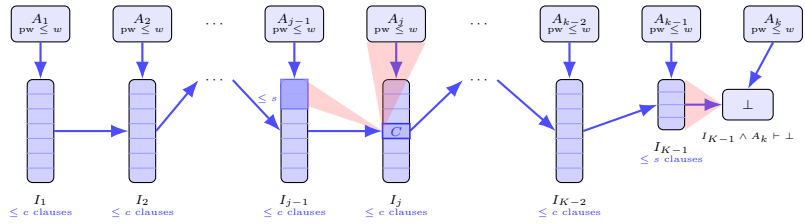
\begin{figure}[t]
    \centering

    \begin{tikzpicture}[
      scale=0.7, transform shape,
      font=\scriptsize,
      chunk/.style={
        draw, rounded corners=2pt,
        minimum width=1.10cm, minimum height=0.68cm,
        fill=blue!10, align=center
      },
      interp/.style={
        draw, rounded corners=2pt,
        minimum width=0.50cm, minimum height=1.95cm,
        fill=blue!18
      },
      slice/.style={draw=blue!40, line width=0.22pt},
      active/.style={fill=blue!35, opacity=0.9, draw=blue!60, line width=0.3pt},
      arr/.style={-Latex, thick, color=blue!70},
      dots/.style={font=\scriptsize},
      cone/.style={fill=red!70, opacity=0.25, draw=none}
    ]
    
    \newcommand{\SliceInterp}[2]{%
      \foreach \t in {1,...,#2} {
        \draw[slice]
          ($(#1.south west)!\t/(#2+1)!(#1.north west)$) --
          ($(#1.south east)!\t/(#2+1)!(#1.north east)$);
      }%
    }
    
    \node[chunk] (A1) {$A_1$\\[-1mm]\tiny pw $\le w$};
    \node[chunk, right=8mm of A1] (A2) {$A_2$\\[-1mm]\tiny pw $\le w$};
    \node[dots,  right=5mm of A2] (D1) {$\cdots$};
    \node[chunk, right=6mm of D1] (Ajm1) {$A_{j-1}$\\[-1mm]\tiny pw $\le w$};
    \node[chunk, right=8mm of Ajm1] (Aj) {$A_j$\\[-1mm]\tiny pw $\le w$};
    \node[dots,  right=7mm of Aj] (D2) {$\cdots$};
    
    \node[chunk, right=8mm of D2]   (Akm2) {$A_{k-2}$\\[-1mm]\tiny pw $\le w$};
    \node[chunk, right=8mm of Akm2] (Akm1) {$A_{k-1}$\\[-1mm]\tiny pw $\le w$};
    \node[chunk, right=8mm of Akm1] (Ak)   {$A_k$\\[-1mm]\tiny pw $\le w$};
    
    \node[interp, below=7mm of A1]   (I1) {};
    \node[interp, below=7mm of A2]   (I2) {};
    \node[dots,  below=7mm of D1]    (Id1) {$\cdots$};
    \node[interp, below=7mm of Ajm1] (Ijm1) {};
    \node[interp, below=7mm of Aj]   (Ij) {};
    \node[dots,  below=7mm of D2]    (Id2) {$\cdots$};
    
    \node[interp, below=7mm of Akm2] (Ikm2) {};
    
    \node[
      draw, rounded corners=2pt,
      minimum width=0.50cm,
      minimum height=0.95cm, 
      fill=blue!18,
      below=7mm of Akm1
    ] (Ikm1) {};
    
    \node[below=1mm of I1]   {$I_1$};
    \node[below=1mm of I2]   {$I_2$};
    \node[below=1mm of Ijm1] {$I_{j-1}$};
    \node[below=1mm of Ij]   {$I_j$};
    \node[below=1mm of Ikm2] {$I_{k-2}$};
    \node[below=1mm of Ikm1] {$I_{k-1}$};
    
    \node[below=3.5mm of Ikm1, font=\tiny, color=blue!80]
    {$\leq s$ clauses};

    \node[below=3.5mm of I1, font=\tiny, color=blue!80]
    {$\leq c$ clauses};
    \node[below=3.5mm of I2, font=\tiny, color=blue!80]
    {$\leq c$ clauses};
    \node[below=3.5mm of Ijm1, font=\tiny, color=blue!80]
    {$\leq c$ clauses};
    \node[below=3.5mm of Ij, font=\tiny, color=blue!80]
    {$\leq c$ clauses};
    \node[below=3.5mm of Ikm2, font=\tiny, color=blue!80]
    {$\leq c$ clauses};
    
    \SliceInterp{I1}{6}
    \SliceInterp{I2}{6}
    \SliceInterp{Ijm1}{6}
    \SliceInterp{Ij}{6}
    \SliceInterp{Ikm2}{6}
    \SliceInterp{Ikm1}{3} 
    
    \draw[arr] (A1.south)   -- (I1.north);
    \draw[arr] (A2.south)   -- (I2.north);
    \draw[arr] (Ajm1.south) -- (Ijm1.north);
    \draw[arr] (Aj.south)   -- (Ij.north);
    
    \draw[arr] (Akm2.south) -- (Ikm2.north);
    \draw[arr] (Akm1.south) -- (Ikm1.north);
    
    \draw[arr] (I1.east)   -- (I2.west);
    \draw[arr] (I2.east)   -- (Id1.west);
    \draw[arr] (Id1.east)  -- (Ijm1.west);
    \draw[arr] (Ijm1.east) -- (Ij.west);
    \draw[arr] (Ij.east)   -- (Id2.west);
    \draw[arr] (Id2.east)  -- (Ikm2.west);
    \draw[arr] (Ikm2.east) -- (Ikm1.west);

    \path (Ijm1.south west) coordinate (SWp);
    \path (Ijm1.south east) coordinate (SEp);
    \path (Ijm1.north west) coordinate (NWp);
    \path (Ijm1.north east) coordinate (NEp);
    
    \path ($(SWp)!5/7!(NWp)$) coordinate (ActLowL);
    \path ($(SEp)!5/7!(NEp)$) coordinate (ActLowR);
    
    \fill[active] (ActLowL) rectangle (NEp);
    
    \node[font=\tiny, color=blue!80, align=right]
      at ($(ActLowL)+(-0.25cm,0.20cm)$) {$\le s$};
    
    \path (Ij.south west) coordinate (SWj);
    \path (Ij.south east) coordinate (SEj);
    \path (Ij.north west) coordinate (NWj);
    \path (Ij.north east) coordinate (NEj);
    
    \path ($(SWj)!3/7!(NWj)$) coordinate (CloL);
    \path ($(SEj)!3/7!(NEj)$) coordinate (CloR);
    \path ($(SWj)!4/7!(NWj)$) coordinate (ChiL);
    \path ($(SEj)!4/7!(NEj)$) coordinate (ChiR);
    
    \fill[blue!25] (CloL) rectangle (ChiR);
    \draw[blue!80, line width=0.6pt] (CloL) rectangle (ChiR);
    \node[font=\scriptsize, color=blue!85] at ($(CloL)!0.5!(ChiR)$) {$C$};
    
    \coordinate (Tip) at ($(CloL)!0.5!(ChiL)$);
    
    \coordinate (BaseAleft)  at (Aj.south west);
    \coordinate (BaseAright) at (Aj.south east);
    \fill[cone] (BaseAleft) -- (BaseAright) -- (Tip) -- cycle;
    
    \coordinate (BaseIupper) at ($(NEp)+(0.00cm,-0.05cm)$);
    \coordinate (BaseIlower) at ($(ActLowR)+(0.00cm,+0.05cm)$);
    \fill[cone] (BaseIupper) -- (BaseIlower) -- (Tip) -- cycle;

    \node[
      draw, rounded corners=2pt, fill=blue!10,
      minimum width=0.95cm, minimum height=0.58cm,
      right=7mm of Ikm1, align=center
    ] (Bot) {$\bot$};
    
    \draw[arr] (Ak.south) -- (Bot.north);
    \draw[arr] (Ikm1.east) -- (Bot.west);
    
    \node[font=\tiny, below=1mm of Bot]
    {$I_{k-1}\wedge A_k \vdash \bot$};

    \fill[cone]
      (Ikm1.north east) --
      (Ikm1.south east) --
      (Bot.west) --
      cycle;
    
    \end{tikzpicture}

    \caption{Structure of a proofdoor. }
\label{fig:proofdoor-structure}
\end{figure}


We now present our main result, which connects small proofdoors to short proofs.

\begin{theorem}[\bf{Proofdoor Theorem}]\label{thm:proofdoor-main}
If $F$ is an unsatisfiable CNF formula that admits a proofdoor with 
parameters $(c,w,s,k)$, where $c$ is the maximum number of clauses in 
each interpolant, $w$ is the maximum pathwidth of each chunk, $s$ is the 
maximum backward dependency of each interpolant clause on the previous 
interpolant, and $k$ is the number of chunks, then $F$ has a resolution 
refutation of size $O(n3^{w+s}kc)$. Furthermore, this refutation respects the proofdoor's cutting partial order.
\end{theorem}
\begin{proof}
Let $F = A_1 \land \cdots \land A_k$ be a proofdoor decomposition with interpolants $I_1, I_2, \cdots I_{k-1}$ satisying the properties of Definition \ref{def:proofdoordecomp}.  We show that each interpolant $I_j$ can be derived from $I_{j-1} \wedge A_j$
by a resolution derivation of size $O(n3^{w+s}\cdot c)$, and that combining these derivations yields a refutation of $F$ of the same asymptotic size.

Fix an index $j \in \{2,\ldots,k-1\}$. Consider an arbitrary clause
$C = (\ell_1 \vee \ell_2 \vee \cdots \vee \ell_t)$ appearing in the interpolant $I_j$.
By the definition of interpolant, we have $(I_{j-1} \wedge A_j) \models C.$
By our assumption we also have $(S(C) \wedge A_j) \models C.$ Equivalently,
if we let $\alpha$ be the partial assignment that falsifies all literals of $C$,
then the restricted formula $(S(C) \wedge A_j)\!\upharpoonright_\alpha$
is unsatisfiable.
Observe that $A_j\!\upharpoonright_\alpha$  has pathwidth at most $w$ with respect to the same path decomposition under restriction.
Moreover, we need only $s$ clauses from the previous interpolant to derive $C$.
So, the restriction $S(C)\!\upharpoonright_\alpha$ contributes at most $s$ clauses. Now, the formula $(S(C) \wedge A_j)\!\upharpoonright_\alpha$ is also a bounded pathwidth formula, as one can create a new path decomposition of
$(S(C) \wedge A_j)\!\upharpoonright_\alpha$ by creating new bags for all the variables
not in $A_j$ but in $S(C)$ after the path decomposition of $A_j$. And we add all the
clauses in $S(C)$ to all the bags, but as the number of clauses is just $s$ this will increase the pathwidth by maximum $s$.
Reversing the restriction $\alpha$ yields a resolution derivation of the clause $C$
from $I_{j-1} \wedge A_j$ of the same size. Since
$(S(C) \wedge A_j)\!\upharpoonright_\alpha$ has pathwidth at most $w$, it has a
resolution refutation of size $O(n3^{w+s})$~\cite{imanishi2017upper}. Thus each clause $C \in I_j$ can be
derived from $I_{j-1} \wedge A_j$ by a resolution derivation of size $O(n3^{w+s})$.
Since $I_j$ contains at most $c$ clauses, the entire interpolant $I_j$ can be derived
from $I_{j-1} \wedge A_j$ by a resolution derivation of size $O(n3^{w+s}\cdot c)$. For $j=1$, we have $A_1 \models C$ directly, and the same pathwidth
argument gives each $C\in I_1$ in size $O(n3^{w})$; the case $j\ge 2$
follows. Combining these derivations for $j = 1,\ldots,k-1$ yields a resolution derivation of
$I_{k-1}$ of size $O((k)\cdot n3^{w+s}\cdot c)$. Continuing with $I_{k-1}$ and $A_k$,
we can derive $\bot$, since the last interpolant has at most $s$ clauses so
$I_{k-1} \land A_k$ has pathwidth at most $w+s$. This final step yields a refutation of
size $O(n3^{w+s})$. Therefore the total proof size is $O(n3^{w+s}kc)$.
\end{proof}

\begin{corollary}
    [\bf{Small Proofdoor Corollary}]\label{thm:smallpfd}
    If $F$ admits small proofdoors, then $F$ has a resolution refutation of polynomial size in $n$.
\end{corollary}

Next, we connect Corollary~\ref{thm:smallpfd} to CDCL. A direct implication of
Atserias, Fichte, and Thurley~\cite{atserias2011clause} framework is that CDCL
with nondeterministic branching and nondeterministic value selection
$p$-simulates general resolution. Combined with Corollary~\ref{thm:smallpfd},
which gives every small-proofdoor formula a resolution refutation of size
polynomial in $n$, this shows that CDCL with nondeterministic branching,
nondeterministic value selection, \textsc{Decision} learning, restart after each
conflict, no clause deletion, and naive BCP refutes every formula admitting small
proofdoors in a polynomial number of conflicts, and hence in polynomial time.

\begin{theorem}[\bf{{CDCL-Proofdoor Theorem}}]
\label{thm:cdcl-proofdoors}
Suppose $F$ admits small proofdoors  via the
partition $F = A_1 \wedge \cdots \wedge A_k$
and interpolants $(I_1,\dots,I_{k-1})$
as in Theorem~\ref{thm:proofdoor-main}. Then there exists a CDCL run
(with non-deterministic value selection, non-deterministic branching,
DECISION learning, restart after each conflict,
no clause deletion, and naive BCP)
that refutes $F$ with at most $\poly(n)$ conflicts.
\end{theorem}

\section{Commutativity of Floating-Point Addition}
\label{sec:fp} 

We now consider the \emph{commutativity of floating-point addition}. We model a floating-point representation inspired from IEEE~754~\cite{ieee754}. A normalized floating-point number is represented as $
x = (-1)^S \cdot (1.F) \cdot 2^{E}$, 
where $S$ is the sign, $E$ the unbiased exponent, and $F$ the fraction bits.  We consider positive normalized inputs only, and use round-to-nearest, ties-to-even. 
We do not model NaNs, infinities, subnormals, or signed zeros. 
Given two inputs $a=(E_a,M_a)$ and $b=(E_b,M_b)$ where $E_a,E_b$ are $m$ bits and  $M_a,M_b$ are $n$ bits, we encode a simplified floating-point addition circuit  inspired by the standard pipeline, consisting of the following stages:

\begin{enumerate}
    \item {Exponent Comparison:} Determine which operand has the larger exponent. 
    A comparator determines whether $E_a > E_b$. It outputs ordering signals $GT$ and $LT$ and selects the larger and smaller exponents using a multiplexer $E_{\text{large}}, E_{\text{small}}$.

    \item {Mantissa Alignment:} Right-shift the smaller mantissa by the exponent difference. This introduces Guard (G), Round (R), and Sticky (S) bits. 
    The exponent difference $\text{Diff} = E_{\text{large}} - E_{\text{small}}$ determines the right shift of the smaller mantissa. The outputs at this stage are the aligned significand $M'_{\text{small}}$ together with guard, round and sticky bits $G, R, S$.

    \item {Significand Addition:} Perform significand addition in extended precision (includes $G,R,S$). 
    The aligned significands are added, producing a sum $\Sigma$ and carry-out $c$.

    \item {Normalization:} If the result overflows the leading 1-position, adjust the exponent and shift the significand accordingly. 
    If $c=1$ then the sum is shifted right by one position and the exponent is incremented. The outputs are the normalized significand $\Sigma_{\text{norm}}$, the updated $G,R,S$ bits $G_{\text{norm}}, R_{\text{norm}}, S_{\text{norm}}$, and the updated exponent $E_{\text{out}}$.

    \item {Rounding:} Use $G,R,S$ bits (round-to-nearest, ties-to-even) to produce the final representable mantissa. 
    A rounding increment $\text{inc}$ is computed from $R_{\text{norm}}, S_{\text{norm}}$, and the least significant bit of $\Sigma_{\text{norm}}$. The final outputs are $\Sigma_{\text{final}}$ and $E_{\text{final}}$.
\end{enumerate}

The complete CNF encoding has size $O(nm+n+m)$.
Detailed gate-level encodings are given in the appendix \ref{app:floap}.

\begin{theorem}[{\bf FP Small Proofdoor Theorem}]
\label{thm:smallpfdfloat}
The CNF encoding of the commutativity of floating-point addition
admits small proofdoors. In particular, there exists a proofdoor
decomposition with parameters $(c,w,s,k) = (O(m+n),\, O(1),\, O(1),\, O(m+n)),$ where $n$ denotes the number of mantissa bits and $m$ denotes the
number of exponent bits.
\end{theorem}

\begin{proof}
    
We construct the usual miter between $L = \mathrm{add}(a,b)$ and $R = \mathrm{add}(b,a)$
and assert that some output bit will differ. We partition the miter encoding into consecutive chunks following the datapath of a standard floating-point adder:
(i) exponent comparison together with exponent/mantissa selection,
(ii) exponent difference,
(iii) mantissa alignment (barrel shifter with GRS logic),
(iv) significand addition,
(v) normalization and exponent update,
(vi) rounding and the final increment/shift.
Each stage is further refined into constant-fan-in slices (bit-cells for add/sub/increment, and multiplexer layers for the shifter), yielding a clause partition
$F = A_1 \wedge \cdots \wedge A_K$.

The interpolants $I_j$ record equalities between the two circuit copies on the wires that feed the next chunk.
At the first stage, symmetry of the comparator under swapping inputs yields mirrored ordering signals satisfying
$\text{EQ}_L = \text{EQ}_R,\ \text{GT}_L = \text{LT}_R,\ \text{LT}_L = \text{GT}_R$.
From these equalities, the exponent-selection multiplexers derive identical
$E_{\text{large}}$ and $E_{\text{small}}$ on both sides, and hence identical exponent difference $\text{Diff}$.
Thus the next interpolant consists of equalities for
$E_{\text{large}}, E_{\text{small}}, \text{Diff}$.

In the alignment stage, equality of $\text{Diff}$ and of the selected mantissas entails equality of the aligned outputs:
$M'_{\text{small},L} = M'_{\text{small},R},\ 
G_L = G_R,\ 
R_L = R_R,\ 
S_L = S_R$.
For arithmetic blocks such as the subtractor and ripple-carry adders, we slice bit-by-bit:
the equality of the $(i+1)$st carry/borrow and the equality of the $i$-th sum/difference bit are entailed by the local CNF of the $i$-th cell together with the equalities of the two input bits and the incoming carry/borrow.
Concretely, $\Sigma_{i,L} = \Sigma_{i,R}$ and $c_{i+1,L} = c_{i+1,R}$
follow from equality of $c_{i,L}$ and the two operand bits.
Thus every newly derived equality clause has a support set $S(C)\subseteq I_{j-1}$ of constant size.
The same holds for the barrel shifter when decomposed by stages:
at each layer, output-wire equality follows from the MUX constraints together with the equality of the select bit and the equalities of the two input wires from the previous layer.

Normalization depends only on the carry-out; since $c_L=c_R$, we obtain equalities for
$\Sigma_{\text{norm}}$, the updated $G_{\text{norm}},R_{\text{norm}},S_{\text{norm}}$, and $E_{\text{out}}$.
Rounding depends only on these wires and the least significant bit of $\Sigma_{\text{norm}}$, yielding equality of
$\Sigma_{\text{final}}$ and $E_{\text{final}}$.
These equalities form the final stage interpolant.

To enforce inequivalence in the miter we introduce error variables
$e_i^M$ and $e_i^E$ encoding disagreements of mantissa and exponent
output bits, and include the clause $\Bigl(\bigvee_{i=0}^{n-1} e_i^M\Bigr)
\;\vee\;
\Bigl(\bigvee_{i=0}^{m-1} e_i^E\Bigr)$. In the final stage the circuit computes $\Sigma_{\text{final}}$ and
$E_{\text{final}}$.
As each output-bit equality
$\Sigma_{\text{final},L}[i]=\Sigma_{\text{final},R}[i]$ or
$E_{\text{final},L}[i]=E_{\text{final},R}[i]$
is derived, the corresponding literal $e_i^M$ or $e_i^E$
is removed from this clause.
Thus the proofdoor carries a single disagreement clause whose length
decreases as output-bit equalities are established, and the final
interpolant consists of a constant number of clauses. Each chunk has bounded clause variable incidence pathwidth (see Appendix \ref{app:floap}).
All conditions of Definition~\ref{def:smallpfd} are satisfied.

\end{proof}

\section{Lower Bounds on Interpolants, Proofdoors, and Partially Ordered Resolution}

Given a proofdoor decomposition, the refutation constructed in
Theorem~\ref{thm:proofdoor-main} derives every clause of each interpolant
$I_1,\dots,I_{k-1}$ in sequence, and each $I_i$ is an interpolant for the split 
$A_1\wedge\cdots\wedge A_i$ versus $A_{i+1}\wedge\cdots\wedge A_k$.
Consequently, a lower bound on the size of \emph{every} interpolant CNF
across a single split will lower-bound the parameter $c$ of every
proofdoor decomposition over that chunking, and hence the size of the
refutations produced by Theorem~\ref{thm:proofdoor-main}, even when a
different decomposition of the same formula admits small proofdoors. In this
section we prove such interpolant size lower bounds for a large class of arithmetic
miter formulas and decompositions, and conclude that proofdoor decompositions containing such cuts yield exponentially large refutations (Theorem~\ref{thm:proofdoor_lb}).

We also show that interpolant lower bounds imply lower bounds for
partially ordered resolution (Lemma~\ref{lem:interpolants}). However, these resolution lower bounds apply to partial orders that are more restrictive than the cutting partial orders given by Definition~\ref{def:proofdoordecomp}, so they do
not apply to the refutations of Theorem~\ref{thm:proofdoor-main}. They may
nevertheless be of independent interest. Whereas feasible
interpolation for unrestricted resolution yields interpolants only as
Boolean circuits~\cite{Krajicek97,Pudlak97}, so that lower bounds
obtained this way are conditional on circuit lower bounds,
Lemma~\ref{lem:interpolants} extracts interpolants in CNF form, where
unconditional exponential bounds are comparatively abundant. As one instance, we use this method in
Theorem~\ref{thm:func-enc} to generalize a partially ordered resolution
lower bound of Janota~\cite{janota2016exponential}.

\subsection{Lower bounds on arithmetic miter interpolants}

A \emph{miter circuit} checks equivalence of two Boolean circuits by computing their outputs on the same inputs and asserting that they differ. Such circuits are standard in equivalence checking and other formal-verification workflows. We will prove size lower bounds for interpolants corresponding to a natural class of arithmetic miter decompositions.

\begin{definition}
A \emph{tree-like arithmetic circuit} is defined as a binary tree graph $T$ whose internal nodes are labeled by $+$ or $\times$, and leaf nodes representing input bit-vectors $\mathbf{x}_0,\mathbf{x}_1,\ldots$. Let $T(n)$ denote the Boolean circuit obtained by instantiating each leaf as an \(n\)-bit input, and replacing the $+$ and $\times$ nodes of $T$ with sufficiently large adder and multiplier circuits such that no overflow occurs. Let $\mathbf{out}_T$ denote the \emph{output} bit-vector of $T(n)$, corresponding to the root node of the underlying graph $T$. Semantically, $\mathbf{out}_T$ represents an algebraic expression of the input bit-vectors $\mathbf{x_0}, \mathbf{x_1}, \ldots$ at the leaf nodes.
\end{definition}

Our arguments do not depend on the particular implementation of the adders and multipliers, only on their functional correctness.

\begin{definition}
An \emph{\(n\)-bit tree-like arithmetic miter circuit} consists of two tree-like arithmetic circuits \(T_1(n)\) and \(T_2(n)\) over the same leaf inputs, together with a Boolean circuit \(E\) that outputs \(1\) iff \(\mathbf{out}_{T_1}\neq \mathbf{out}_{T_2}\). In the instances of interest, \(T_1\) and \(T_2\) compute the same arithmetic expression, so the overall miter is unsatisfiable.
\end{definition}

\begin{figure}
    \centering
    \includegraphics[width=0.7\linewidth]{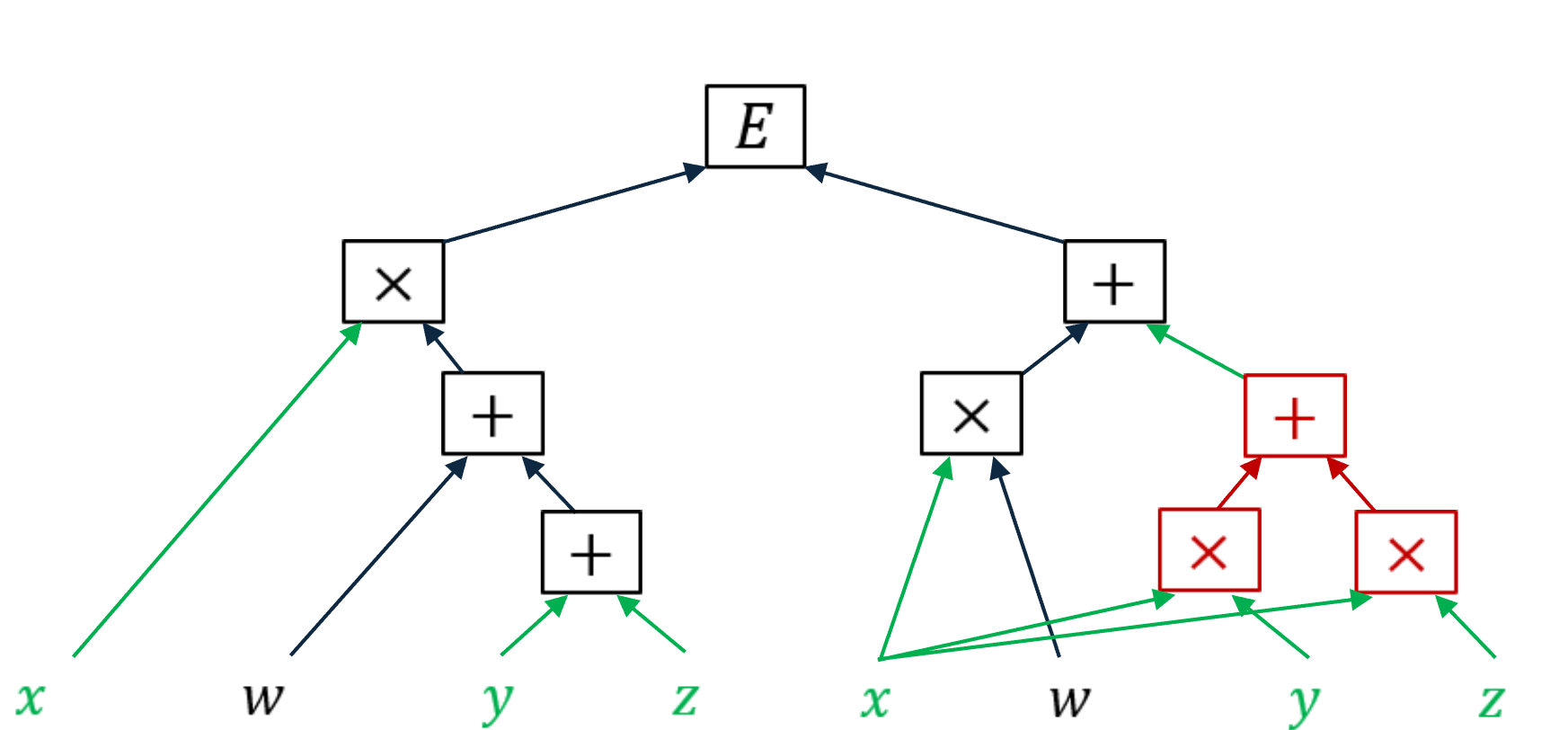}
    \caption{A tree-like arithmetic miter circuit encoding the algebraic inequality $\mathbf{x}(\mathbf{w}+\mathbf{y}+\mathbf{z}) \neq \mathbf{xw}+(\mathbf{xy}+\mathbf{xz})$. If $v$ is chosen to be the (total) $+$ node that outputs $\mathbf{xy}+\mathbf{xz}$, then the variables belonging to $\Vbefore$ are colored in black, the variables belonging to $\Vafter$ are colored in red, and the variables belonging to $\Vshared$ are colored in green.}
    \label{fig:miter}
\end{figure}

See Figure~\ref{fig:miter} for an example of a tree-like arithmetic miter circuit. Most of the miter circuits considered in the literature are tree-like. For instance,~\cite{beame2019toward} studied tree-like arithmetic miter circuits corresponding to algebraic equalities such as $\mathbf{xy} = \mathbf{yx}$ or $\mathbf{x}(\mathbf{y}+\mathbf{z}) = \mathbf{xy} + \mathbf{xz}$.

Our lower bounds will apply to interpolants corresponding to the following class of arithmetic miter decompositions.

\begin{definition}[{\bf Cut at a total node}]\label{def:total_node_cut}
Let $\Phi$ be the CNF encoding of an $n$-bit tree-like arithmetic miter
circuit, and let $v$ be an internal $+$ or $\times$ node of $T_1$ or
$T_2$. Define $\Vafter$ to be the set of circuit variables (i.e.,
excluding leaf input variables) contained in the subtree rooted at $v$;
let $\Vshared$ contain the output bit-vector of $v$ together with the
leaf input bit-vectors contained in that subtree; and let $\Vbefore$ be
the set of all remaining variables. No
clause of $\Phi$ contains both a $\Vbefore$-variable and a
$\Vafter$-variable, so $\Phi$ can be written as a split formula
\[
\Phi \;=\; A(\Vbefore,\Vshared)\ \wedge\ B(\Vafter,\Vshared),
\]
which we call the \emph{cut at $v$}.

We say that $v$ is a \emph{total node} if there exists a restriction $r$
to the leaf input variables in the subtree rooted at $v$ such that, in
the restricted circuit:
\begin{enumerate}
\item the left child of $v$ is fixed to the same value as some unassigned
$n$-bit leaf input vector $\mathbf a$;
\item the right child of $v$ is fixed to the same value as some
unassigned $n$-bit leaf input vector $\mathbf b$ that is distinct from
$\mathbf a$;
\item for every $\times$-node on the path from $v$ to the root, the
node's off-path input is not forced to $0$ by $r$.
\end{enumerate}
\end{definition}

For example, in the miter circuit corresponding to $\mathbf{xy} = \mathbf{yx}$, both of the $\times$ nodes are total nodes. For deeper circuits, we can often obtain a suitable restriction by recursively setting leaf inputs for $+$ nodes to $0$ and inputs for $\times$ nodes to $1$. For example, if $v$ is the $+$ node for the expression $\mathbf{xy}+\mathbf{xz}$, then the restriction $\mathbf{x}=1$ shows that $v$ is a total node. If $v$ is the $\times$ node for the expression $(\mathbf{x}+\mathbf{y})(\mathbf{x}+\mathbf{z})$, then the restriction $\mathbf{x}=0$ shows that $v$ is a total node. If $v$ is the $\times$ node for the expression $\mathbf{x}*\mathbf{x}$, it is \emph{not} a total node as its left and right inputs cannot have different values.

We will obtain our lower bounds by showing that interpolants over a cut at a total node $v$ contain the following Boolean functions, which have exponential CNF and DNF lower bounds.

\begin{definition}\label{def:graph_functions}
For \(n\in\mathbb{N}\), let \(\mathbf{x}=(x_0,\ldots,x_{n-1})\) and \(\mathbf{y}=(y_0,\ldots,y_{n-1})\) be \(n\)-bit inputs.
Write \(x\) and \(y\) for the integers they encode in binary. Define:
\[
\MultGraph(n)(\mathbf{x},\mathbf{y},\mathbf{z})=1 \iff x\cdot y=z,
\]
where \(\mathbf{z}=(z_0,\ldots,z_{2n-1})\) encodes a \(2n\)-bit integer \(z\).
\[
\AddGraph(n)(\mathbf{x},\mathbf{y},\mathbf{z})=1 \iff x+y=z,
\]
where \(\mathbf{z}=(z_0,\ldots,z_n)\) encodes an \((n+1)\)-bit integer \(z\).
\end{definition}

\begin{theorem}[Buss~\cite{Buss1992GraphMultCounting}]\label{thm:multgraph_lower_bound}
\MultGraph(n) requires exponentially large constant-depth circuits.
\end{theorem}

Since CNFs and DNFs are depth-\(2\) circuits, Buss's lower bound (Theorem~\ref{thm:multgraph_lower_bound}) implies exponential size lower bounds for both CNFs and DNFs computing \MultGraph.

In contrast, the addition function, and by extension \AddGraph, has polynomial size depth $3$ Boolean circuits~\cite{Parberry1994CircuitComplexityNeuralNetworks}. Yet at depth $2$, the addition function requires Boolean circuits of size at least $2^{n-1}$~\cite{Parberry1994CircuitComplexityNeuralNetworks}. To our knowledge, whether \AddGraph~requires exponentially large CNFs or not is an open question (see Conjecture~\ref{conj:addgraph}).

In the next theorem we prove a direct $2^n$ DNF size lower bound on both $\AddGraph(n)$ and $\MultGraph(n)$, which implies that their negations require CNFs of size $2^n$.

\begin{theorem}
\label{thm:depth2}
Let $D$ be a DNF computing either \MultGraph(n) or \AddGraph(n). Then $D$ has at least $2^n$ terms.
\end{theorem}

\begin{proof}
Both $\MultGraph(n)$ and $\AddGraph(n)$ contain the equality subfunction
$$
f_{\EQ_n}(x,z) \;=\; \bigwedge_{i=0}^{n-1}(x_i=z_i).
$$
Indeed, for $\MultGraph(n)$, restricting $y=1$ and $z_n=\cdots=z_{2n-1}=0$ yields $x=z$ on the low $n$ bits; for $\AddGraph(n)$, restricting $y=0$ and $z_n=0$ yields the same subfunction. The theorem then follows from the following lemma. The proof uses a standard DNF lower bound argument; we include it in Appendix~\ref{app:EQ_lower_bound}.

\begin{lemma}
\label{EQ_lower_bound}
$f_{\EQ_n}(x,z)$ requires DNFs of at least $2^n$ terms.
\end{lemma}
\end{proof}

\begin{theorem}[{\bf Miter Interpolant Lower Bound}]
\label{thm:miter_lower_bound}
Let $\Phi = A \wedge B$ be the cut at a total node $v$ of an $n$-bit
tree-like arithmetic miter circuit (Definition~\ref{def:total_node_cut}).
Then every interpolant from $A$ to $B$ requires CNFs of size at least
$2^n$.
\end{theorem}

\begin{proof}
Let $\mathbf{v}_L,\mathbf{v}_R$ and $\mathbf{v}_{\mathrm{out}}$ denote the left input, right input, and output bit-vectors of the total node $v$. Since $v$ is a total node, there exists a restriction $r$ of the leaf input variables in the subtree of $v$ such that the restricted circuit constrains $\mathbf v_L=\mathbf a$ and $\mathbf v_R=\mathbf b$ for distinct leaf inputs $\mathbf a,\mathbf b\in V_{\mathrm{shared}}$, and such that this restriction does not force a $0$-input into any $\times$-gate on the path from $v$ to the root. The rest of this proof gives a lower bound for the restricted formula $\Phi \upharpoonright_{r}$.

Set the Boolean function $f(\mathbf a, \mathbf b, \mathbf{v}_{\mathrm{out}})$ to be either $\AddGraph(n)$, if $v$ is a $+$ node, or $\MultGraph(n)$, if $v$ is a $\times$ node.
Let $I$ be any interpolant function from $A(\Vbefore,\Vshared)$ to $B(\Vafter,\Vshared)$. We now show the key claim of this proof, that $I=\neg f$ under the restriction $r$, i.e. $\neg f$ is a subfunction of $I$. Once this claim has been proven, the theorem follows from the DNF lower bound on $f$ given in Theorem~\ref{thm:depth2}.

Suppose that $\rho$ is an assignment to the remaining shared variables $\Vshared = \{\mathbf a, \mathbf b, \mathbf{v}_{\mathrm{out}}\}$ such that
$\neg f(\mathbf a, \mathbf b, \mathbf{v}_{\mathrm{out}})=0$ (i.e. $\rho$ assigns $\mathbf{v}_{\mathrm{out}}$ correctly with respect to $\mathbf{a},\mathbf{b}$).
Then the clauses in $B$ are satisfied by propagating the leaf input values up the subtree rooted at $v$,
and hence $I(\rho) = 0$.

Conversely, suppose that $\rho$ is an assignment to $\Vshared$ such that
$\neg f(\mathbf a, \mathbf b, \mathbf{v}_{\mathrm{out}})=1$ (i.e. $\rho$ assigns $\mathbf{v}_{\mathrm{out}}$ incorrectly with respect to $\mathbf{a},\mathbf{b}$).
Then the clauses $A(\Vbefore,\Vshared)$ are satisfied by assigning the remaining leaf input bit-vectors in $\Vbefore$ to the value $1$, then propagating all leaf inputs and $\mathbf{v}_{\mathrm{out}}$ through the rest of the circuit corresponding to the clauses $A(\Vbefore,\Vshared)$. Because the restriction $r$ does not force a $0$-input into any $\times$-gate on the path from $v$ to the root, each $\times$ and $+$ operation on this path is injective. Therefore the sub-circuit $T_1(n)$ or $T_2(n)$ containing $v$ will output the incorrect value while the other sub-circuit outputs the correct value, and hence the values $\mathbf{out}_{T_1}$ and $\mathbf{out}_{T_2}$ disagree, satisfying the subcircuit $E$. As $A(\Vbefore,\Vshared)$ is satisfiable, $I(\rho) = 1$.
\end{proof}

With a symmetric argument, we can obtain a size lower bound for interpolants going in the reverse direction, from $B(\Vafter,\Vshared)$ to $A(\Vbefore,\Vshared)$, by using the CNF lower bound on $\MultGraph(n)$ implied by Theorem~\ref{thm:multgraph_lower_bound}.

\begin{theorem}
\label{thm:rev_ordering_lower_bound}
If $v$ is a total $\times$ node, then with the same assumptions as in Theorem~\ref{thm:miter_lower_bound}, interpolant CNFs from $B(\Vafter,\Vshared)$ to $A(\Vbefore,\Vshared)$ must be exponentially large.
\end{theorem}

We next state a conjecture that, if proven, would extend Theorem~\ref{thm:rev_ordering_lower_bound} to the case where $v$ is a total $+$ node.

\begin{conjecture}
\label{conj:addgraph}
    CNFs for \AddGraph(n) must be exponentially large.
\end{conjecture}

We now return to the connection with proofdoors described at the start
of this section. Say that a proofdoor decomposition
$\Phi = A_1 \wedge \cdots \wedge A_k$ \emph{contains the cut at $v$ with
prefix $A$} if for some index $i$ the chunks $A_1, \ldots, A_i$ consist
exactly of the clauses of $A(\Vbefore,\Vshared)$; define \emph{contains
the cut at $v$ with prefix $B$} symmetrically.

\begin{theorem}[{\bf Proofdoor Lower Bound}]\label{thm:proofdoor_lb}
Let $\Phi = A \wedge B$ be the cut
at a total node $v$ of an $n$-bit tree-like arithmetic miter circuit
(Definition~\ref{def:total_node_cut}). Every proofdoor decomposition of
$\Phi$ that contains the cut at $v$ with prefix $A$ requires an
interpolant $I_i$ with at least $2^n$ clauses; hence its parameter $c$ (Definition~\ref{def:smallpfd}) is
at least $2^n$, and the refutation constructed by
Theorem~\ref{thm:proofdoor-main} from any such decomposition has size at
least $2^n$. If $v$ is a $\times$ node, every proofdoor decomposition containing the
cut at $v$ with prefix $B$ likewise requires an exponentially large
interpolant $I_i$, hence exponential $c$ and refutation size; if $v$ is
a $+$ node, the same holds assuming Conjecture~\ref{conj:addgraph}.
\end{theorem}

\begin{proof}
Suppose that the proofdoor decomposition contains the cut at $v$ with prefix $A$, so
$A_1 \wedge \cdots \wedge A_i = A$ and
$A_{i+1} \wedge \cdots \wedge A_k = B$ as sets of clauses. Then $I_i$ is an
interpolant from $A$ to $B$ over $\Vshared$, and by
Theorem~\ref{thm:miter_lower_bound} its CNF has at least $2^n$ clauses. So this proofdoor decomposition has parameter $c \geq 2^n$. Since the refutation of
Theorem~\ref{thm:proofdoor-main} derives every clause of $I_i$, it has
size at least $2^n$. With prefix $B$, the same argument shows $I_i$ is
an interpolant from $B$ to $A$, and the bounds follow from
Theorem~\ref{thm:rev_ordering_lower_bound} and
Conjecture~\ref{conj:addgraph}.
\end{proof}

Appendix~\ref{app:miterasthm1} gives a small proofdoor for the miter circuit $\mathbf{xy}=\mathbf{yx}$, corresponding to the polynomial size ordered resolution proof given in~\cite{beame2019toward}. Consistent with Theorem~\ref{thm:proofdoor_lb}, this small proofdoor does not contain the cut at either $\times$-node of the miter circuit. Thus, for the same formula, the choice of decomposition is the difference between the polynomial-size refutations of Corollary~\ref{thm:smallpfd} and the exponential refutations of Theorem~\ref{thm:proofdoor_lb}.

\subsection{Lower bounds for partially ordered resolution}

The following lemma yields partially ordered resolution lower bounds from interpolant CNF size lower bounds.

\begin{lemma}[{\bf Interpolants for partially ordered resolution}]
\label{lem:interpolants}
Let $\Phi$ be an unsatisfiable CNF formula written as $\Phi = A(\Vbefore,\Vshared) \wedge B(\Vafter,\Vshared)$, where $\Vbefore$ and $\Vafter$ are disjoint sets of variables. Let $\prec$ be the partial ordering $\Vbefore \prec \Vafter \cup \Vshared$.

From a $\prec$-ordered resolution refutation of $\Phi$ of size $s$, we can construct a CNF formula $C$ of size at most $s$ that is an interpolant from $A(\Vbefore,\Vshared)$ to $B(\Vafter,\Vshared)$.
\end{lemma}

\begin{proof}
    Let $\pi$ be a $\prec$-ordered resolution refutation of $\Phi$ of size at most $s$. Each leaf-to-root path $P$ in $\pi$ that resolves on at least one $\Vbefore$-variable contains a ``boundary clause'' $C_P$ that is produced from the last resolution step on a $\Vbefore$-variable on the path $P$ (see Figure~\ref{fig:interpolant_thm}). The remaining clauses on the path $P$ must result from resolving on variables from $\Vafter$ or $\Vshared$. Therefore $C_P$ cannot contain variables from $\Vbefore$: if it did contain a variable $v \in \Vbefore$, this variable cannot be eliminated in any later resolution steps of the path $P$. Yet the path $P$ ends in the empty clause, which is a contradiction. The clause $C_P$ also cannot contain variables from $\Vafter$, since it was derived by resolving only on $\Vbefore$ variables, meaning it was derived only using clauses from $A(\Vbefore, \Vshared)$.
    
    Consider the set $C$ containing the boundary clauses $C_P$ for every leaf-to-root path $P$ in $\pi$. We will show that this set $C$ forms a CNF for an interpolant function $C(\Vshared)$ from $A$ to $B$. The implication $A(\Vbefore, \Vshared) \models C(\Vshared)$ follows from the soundness of the resolution rule. The unsatisfiability of $C(\Vshared) \land \neg B(\Vafter,\Vshared)$ follows from observing that the clauses $C\cup B$ form a cut in the proof DAG of $\pi$, and is therefore an unsatisfiable set of clauses.
\end{proof}

\begin{figure}
    \centering
    \includegraphics[width=0.7\linewidth]{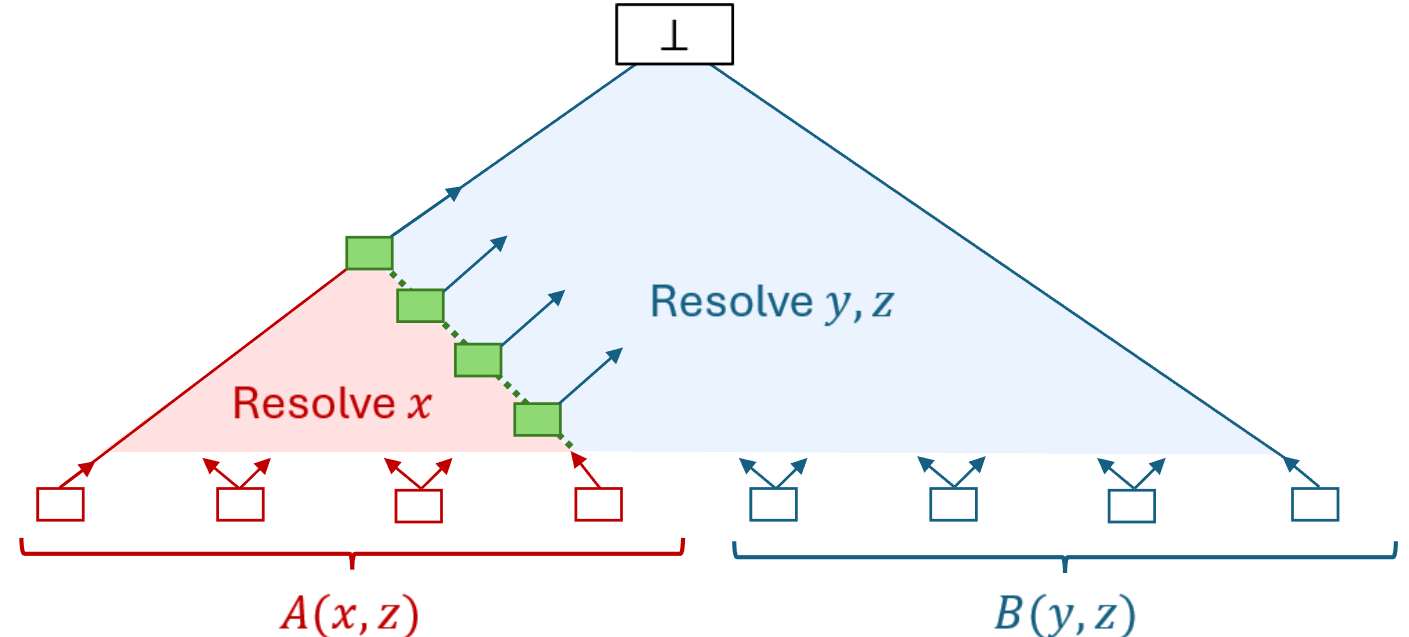}
    \caption{A partially ordered resolution proof for a formula $\Phi = A(X,Z) \wedge B(Y,Z)$. The green boxes along the dashed line are the ``boundary clauses''. These boundary clauses only contain $Z$-variables, and together they form an interpolant from $A(X,Z)$ to $B(Y,Z)$.}
    \label{fig:interpolant_thm}
\end{figure}

\begin{corollary}
\label{cor:interpolants-janota}
Any resolution refutation $\pi$ of $\Phi = A(\Vbefore,\Vshared) \wedge B(\Vafter,\Vshared)$ that respects the partial order $(\Vbefore\cup \Vafter) \prec \Vshared$ also respects the partial order $\Vbefore \prec (\Vafter \cup \Vshared)$. Hence, Lemma~\ref{lem:interpolants} also holds for the partial order $(\Vbefore\cup \Vafter) \prec \Vshared$.
\end{corollary}

\begin{proof}
It suffices to show that in the resolution refutation $\pi$, no leaf-to-root path $P$ resolves on a $\Vafter$ variable before an $\Vbefore$ variable. This is true because the set of clauses containing $\Vbefore$ variables is disjoint from the set of clauses containing $\Vafter$ variables. So a leaf-to-root path $P$ that has resolved on a $\Vafter$ variable cannot later resolve on an $\Vbefore$ variable before it has resolved on a shared $\Vshared$ variable.
\end{proof}

\begin{theorem}\label{thm:miter_po_res}
Let $\Phi = A(\Vbefore,\Vshared) \wedge B(\Vafter,\Vshared)$ be the cut at
a total node $v$ of an $n$-bit tree-like arithmetic miter circuit
(Definition~\ref{def:total_node_cut}). Then any resolution refutation of
$\Phi$ respecting the partial order $\Vbefore \prec (\Vafter \cup \Vshared)$
has size at least $2^n$. If $v$ is a $\times$ node, any refutation
respecting the partial order $\Vafter \prec (\Vbefore \cup \Vshared)$ is
exponentially large; if $v$ is a $+$ node, the same holds assuming
Conjecture~\ref{conj:addgraph}.
\end{theorem}
\begin{proof}
By Lemma~\ref{lem:interpolants}, a refutation of size $s$ respecting the partial order $\Vbefore \prec (\Vafter \cup \Vshared)$ yields an
interpolant CNF from $A(\Vbefore,\Vshared)$ to $B(\Vafter,\Vshared)$ of size at most $s$. Since every such CNF has size at least $2^n$ by
Theorem~\ref{thm:miter_lower_bound}, $s \geq 2^n$. Applying
Lemma~\ref{lem:interpolants} with the roles of $A$ and $B$ exchanged, the
same reasoning with Theorem~\ref{thm:rev_ordering_lower_bound} and Conjecture~\ref{conj:addgraph} yields the
lower bounds for refutations respecting
$\Vafter \prec (\Vbefore \cup \Vshared)$.
\end{proof}

We remark that the ordering used in the polynomial size ordered resolution proof of $\mathbf{xy}=\mathbf{yx}$ given in~\cite{beame2019toward} does not fall under either of the partial orderings specified by Theorem~\ref{thm:miter_po_res}. Instead, the proof in~\cite{beame2019toward} follows an ordering that interleaves between the two multiplier circuits of the miter, evading these lower bounds.

\subsection{Lower bounds for \textsc{Function Encoding}}
\label{app:fn_encoding}

Janota~\cite{janota2016exponential} proved an exponential partially ordered resolution lower bound for a type of formula family he termed \textsc{Function Encoding}; Lemma~\ref{lem:interpolants} generalizes this result.

For a Boolean function $f: X \rightarrow \{0,1\}$, consider the contradiction $f(X) \wedge \neg f(X)$. Encode $f(X)$ and $\neg f(X)$ as two separate CNF formulas $F^+(T^+,X)$ and $F^-(T^-,X)$, where $T^+$ and $T^-$ are disjoint sets of fresh auxiliary variables. The encoding is valid if for every assignment $\tau$ to $X$, exactly one of $F^{+}(T^{+},\tau)$ and $F^{-}(T^{-},\tau)$ is satisfiable, namely $F^{+}$ when $f(\tau)=1$ and $F^{-}$ when $f(\tau)=0$.

Janota~\cite{janota2016exponential} considered the case where $f(X)$ is the parity function $f_\oplus(x_1, \ldots, x_n) = x_1 \oplus \cdots \oplus x_n$. He showed that for a standard CNF encoding of $f_\oplus(X)$ that produces the formula $F^+ \wedge F^-$, any resolution refutation respecting the partial ordering $(T^+ \cup T^-) \prec X$ is exponentially large in $n$. 

Lemma~\ref{lem:interpolants} immediately generalizes this result to any valid encoding of any Boolean function $f$ that requires a large CNF (e.g. parity or majority). Our lower bound also applies to a less restrictive partial order (see Corollary~\ref{cor:interpolants-janota}).

\begin{theorem}
\label{thm:func-enc}
If every CNF representing $f(X)$ has size at least $s$, then any partially ordered resolution refutation of
$F^+(T^+,X)\wedge F^-(T^-,X)$ respecting the partial order $(T^+)\prec(T^-\cup X)$ has size at least $s$.
\end{theorem}

\begin{proof}
Observe that $f(X)$ is the unique interpolant from $F^+(T^+, X)$ to $F^-(T^-,X)$. The theorem then follows directly from Lemma~\ref{lem:interpolants}.
\end{proof}

\section{Undecidability of the Short Proof Family Problem}

A natural question motivated by our work is whether one can completely 
characterize all families of formulas that are easy for CDCL. We show 
that this question cannot be answered in general: even deciding whether 
a given computable family of formulas admits polynomial-size resolution 
proofs is undecidable. This establishes a fundamental barrier to any 
complete separation of easy instances from hard ones, and explains why 
any parameters such as proofdoors cannot give a complete characterization.

\begin{definition}[Family of CNF formulas]
A \emph{family of CNF formulas} is a computable sequence $\mathcal{F}=(F_i)_{i\in\mathbb{N}}$
of CNF formulas, i.e., there exists a computable $G$ with $G(i)=F_i$.
\end{definition}

\begin{definition}[Family Proof Size Problem (FPSP)]
A family $\mathcal{F}=\{F_i\}$ of unsatisfiable CNF formulas is a
\emph{polynomial-proof family (in Resolution)} if there exists a polynomial $p$
such that every $F_i$ has a resolution refutation of size at most $p(|F_i|)$.
Otherwise, it is a \emph{hard-proof} family.

The \emph{Family Proof Size Problem (FPSP)} asks: given an algorithm $G$
that generates a family $\mathcal{F}_G=\{G(i)\}_{i\in\mathbb{N}}$
of unsatisfiable CNF formulas, determine whether $\mathcal{F}_G$ is a
polynomial-proof family or a hard-proof family in the Resolution proof system.
\end{definition}

\begin{theorem}[{\bf FPSP Undecidability Theorem}]\label{thm:fpsp}
The Family Proof Size Problem (FPSP) is undecidable.
\end{theorem}





    

\begin{proof}
We reduce the Halting Problem~\cite{turing1936computable} to FPSP. 
Given a Turing Machine $M$ and input $x$, the reduction constructs 
the following algorithm $G$. We run $M$ on $x$ for $t$ steps, with $M$ and $x$ hardcoded.

$$G(t) := \begin{cases} \bot & \text{if } M \text{ has not halted on } x \text{ within } t \text{ steps,} \\ \mathrm{PHP}_t & \text{otherwise.} \end{cases}$$

\noindent Since simulating $M$ on $x$ for $t$ steps is decidable, $G$ is computable. 
If $M$ never halts on $x$, then $G(t) = \bot$ for all $t$, and the family 
$\{G(t)\}$ has constant-size proofs. If $M$ halts on $x$ at some step $T$, 
then $G(t) = \mathrm{PHP}_t$ for all $t \geq T$, which requires resolution 
refutations of size $2^{\Omega(t)}$~\cite{haken1985intractability}. 
Therefore deciding FPSP would decide the Halting Problem. 
FPSP is undecidable.
\end{proof}

\section{Conclusion}

In this paper, we proposed proofdoors as a proof-theoretic framework for understanding the efficiency of CDCL SAT solvers on structured verification instances. We showed that small proofdoors guarantee polynomial-size resolution refutations and can be exploited by suitable CDCL runs to obtain polynomial-time refutations. As a concrete application, we proved that floating-point addition commutativity formulas admit small proofdoors. Finally, we identified inherent limitations of the framework and showed that the choice of decomposition can be the difference between short and exponential proofs. As a byproduct, we found a new method for proving partially ordered resolution lower bounds. We believe Proofdoors may offer a useful lens for studying other structured industrial SAT instances. Extending the framework to more general forms of decomposition, for instance, tree-like decompositions is a natural and promising direction. Additionally, the CDCL-Proofdoor Theorem currently relies on non-determinism; reducing this reliance while moving toward more practical branching heuristics such as VSIDS is another worthwhile open problem.



\newpage
\bibliography{lipics-v2021-sample-article}

@inproceedings{mitchell1992hard,
  title={Hard and easy distributions of {SAT} problems},
  author={Mitchell, David and Selman, Bart and Levesque, Hector and others},
  booktitle={Aaai},
  volume={92},
  pages={459--465},
  year={1992}
}

@inproceedings{alekhnovich2002satisfiability,
  title={{Satisfiability, branch-width and Tseitin tautologies}},
  author={Alekhnovich, Michael and Razborov, Alexander A},
  booktitle={The 43rd Annual IEEE Symposium on Foundations of Computer Science, 2002. Proceedings.},
  pages={593--603},
  year={2002},
  organization={IEEE}
}

@techreport{ieee754,
  author = "{IEEE}",
  title  = "{IEEE Standard for Floating-Point Arithmetic}",
  institution = "IEEE",
  number = "IEEE Std 754-2019",
  year   = 2019,
  doi    = "10.1109/IEEESTD.2019.8766229"
}

@inproceedings{kilby2005backbones,
  title={Backbones and backdoors in satisfiability},
  author={Kilby, Philip and Slaney, John and Thi{\'e}baux, Sylvie and Walsh, Toby and others},
  booktitle={AAAI},
  volume={5},
  pages={1368--1373},
  year={2005}
}

@inproceedings{newsham2014impact,
  title={Impact of community structure on {SAT} solver performance},
  author={Newsham, Zack and Ganesh, Vijay and Fischmeister, Sebastian and Audemard, Gilles and Simon, Laurent},
  booktitle={International Conference on Theory and Applications of Satisfiability Testing},
  pages={252--268},
  year={2014},
  organization={Springer}
}

@inproceedings{zulkoski2018effect,
  title={The effect of structural measures and merges on {SAT} solver performance},
  author={Zulkoski, Edward and Martins, Ruben and Wintersteiger, Christoph M and Liang, Jia Hui and Czarnecki, Krzysztof and Ganesh, Vijay},
  booktitle={International Conference on Principles and Practice of Constraint Programming},
  pages={436--452},
  year={2018},
  organization={Springer}
}

@article{haken1985intractability,
  title={The intractability of resolution},
  author={Haken, Armin},
  journal={Theoretical computer science},
  volume={39},
  pages={297--308},
  year={1985},
  publisher={Elsevier}
}

@article{beame2019toward,
  title={Toward verifying nonlinear integer arithmetic},
  author={Beame, Paul and Liew, Vincent},
  journal={Journal of the ACM (JACM)},
  volume={66},
  number={3},
  pages={1--30},
  year={2019},
  publisher={ACM New York, NY, USA}
}

@article{janota2016exponential,
  title={On exponential lower bounds for partially ordered resolution},
  author={Janota, Mikol{\'a}{\v{s}}},
  journal={Journal on Satisfiability, Boolean Modelling and Computation},
  volume={10},
  number={1},
  pages={1--9},
  year={2016},
  publisher={SAGE Publications Sage UK: London, England}
}

@article{craig1957three,
  title={Three uses of the {Herbrand}-{Gentzen} theorem in relating model theory and proof theory},
  author={Craig, William},
  journal={The Journal of Symbolic Logic},
  volume={22},
  number={3},
  pages={269--285},
  year={1957},
  publisher={Cambridge University Press}
}

@inproceedings{cheeseman1991really,
  title={Where the really hard problems are},
  author={Cheeseman, Peter},
  booktitle={International Joint Conference on Artificial Intelligence},
  year={1991}
}

@article{biere2009bounded,
  title={Bounded model checking.},
  author={Biere, Armin and Cimatti, Alessandro and Clarke, Edmund M and Strichman, Ofer and Zhu, Yunshan},
  journal={Handbook of satisfiability},
  volume={185},
  number={99},
  pages={457--481},
  year={2009}
}

@article{prasad2001combinational,
  title={Why is combinational {ATPG} efficiently solvable for practical {VLSI} circuits?},
  author={Prasad, Mukul R and Chong, Philip and Keutzer, Kurt},
  journal={Journal of Electronic Testing},
  volume={17},
  number={6},
  pages={509--527},
  year={2001},
  publisher={Springer}
}

@inproceedings{bodlaender1991approximating,
  title={Approximating treewidth, pathwidth, and minimum elimination tree height},
  author={Bodlaender, Hans L and Gilbert, John R and Hafsteinsson, Hj{\'a}lmt{\`y}r and Kloks, Ton},
  booktitle={International Workshop on Graph-Theoretic Concepts in Computer Science},
  pages={1--12},
  year={1991},
  organization={Springer}
}

@article{turing1936computable,
  title={On computable numbers, with an application to the {Entscheidungsproblem}},
  author={Turing, Alan Mathison and others},
  journal={J. of Math},
  volume={58},
  number={345-363},
  pages={5},
  year={1936},
  publisher={Wiley Online Library}
}

@book{biere2009handbook,
  title={Handbook of satisfiability},
  author={Biere, Armin and van Maaren, Hans and Walsh, Toby},
  year={2009},
  publisher={SAGE Publications Limited}
}

@inproceedings{mcmillan2003interpolation,
  title={Interpolation and {SAT}-based model checking},
  author={McMillan, Kenneth L},
  booktitle={International Conference on Computer Aided Verification},
  pages={1--13},
  year={2003},
  organization={Springer}
}

@article{robertson1983graphminors,
  author = {{Neil Robertson and Paul D. Seymour}},
  title = {{Graph Minors I}: Excluding a Forest},
  journal = {Journal of Combinatorial Theory, Series B},
  year = {1983}
}

@article{atserias2011clause,
  title={Clause-learning algorithms with many restarts and bounded-width resolution},
  author={Atserias, Albert and Fichte, Johannes Klaus and Thurley, Marc},
  journal={Journal of Artificial Intelligence Research},
  volume={40},
  pages={353--373},
  year={2011}
}

@inproceedings{pipatsrisawat2009power,
  title={On the power of clause-learning {SAT} solvers with restarts},
  author={Pipatsrisawat, Knot and Darwiche, Adnan},
  booktitle={International Conference on Principles and Practice of Constraint Programming},
  pages={654--668},
  year={2009},
  organization={Springer}
}

@article{clarke2001bounded,
  title={Bounded model checking using satisfiability solving},
  author={Clarke, Edmund and Biere, Armin and Raimi, Richard and Zhu, Yunshan},
  journal={Formal methods in system design},
  volume={19},
  number={1},
  pages={7--34},
  year={2001},
  publisher={Springer}
}

@incollection{tseitin1983complexity,
  title={On the complexity of derivation in propositional calculus},
  author={Tseitin, Grigori S},
  booktitle={Automation of reasoning: 2: Classical papers on computational logic 1967--1970},
  pages={466--483},
  year={1983},
  publisher={Springer}
}

@inproceedings{10.1145/800157.805047,
author = {Cook, Stephen A.},
title = {The complexity of theorem-proving procedures},
year = {1971},
isbn = {9781450374644},
publisher = {Association for Computing Machinery},
address = {New York, NY, USA},
url = {https://doi.org/10.1145/800157.805047},
doi = {10.1145/800157.805047},
booktitle = {Proceedings of the Third Annual ACM Symposium on Theory of Computing},
pages = {151–158},
numpages = {8},
location = {Shaker Heights, Ohio, USA},
series = {STOC '71}
}

@inproceedings{kautz1992planning,
  title={{Planning as Satisfiability.}},
  author={Kautz, Henry A and Selman, Bart and others},
  booktitle={ECAI},
  volume={92},
  pages={359--363},
  year={1992}
}

@inproceedings{xie2005saturn,
  title={Saturn: A {SAT}-based tool for bug detection},
  author={Xie, Yichen and Aiken, Alex},
  booktitle={International Conference on Computer Aided Verification},
  pages={139--143},
  year={2005},
  organization={Springer}
}

@book{davis1958feasible,
  title={Feasible computational methods in the propositional calculus},
  author={Davis, Martin and Putnam, Hilary},
  year={1958},
  publisher={Rensselaer Polytechnic Institute, Research Division}
}

@InProceedings{vinyals_et_al:LIPIcs.SAT.2023.27,
  author =	{Vinyals, Marc and Li, Chunxiao and Fleming, Noah and Kolokolova, Antonina and Ganesh, Vijay},
  title =	{{Limits of CDCL Learning via Merge Resolution}},
  booktitle =	{26th International Conference on Theory and Applications of Satisfiability Testing (SAT 2023)},
  pages =	{27:1--27:19},
  series =	{Leibniz International Proceedings in Informatics (LIPIcs)},
  ISBN =	{978-3-95977-286-0},
  ISSN =	{1868-8969},
  year =	{2023},
  volume =	{271},
  editor =	{Mahajan, Meena and Slivovsky, Friedrich},
  publisher =	{Schloss Dagstuhl -- Leibniz-Zentrum f{\"u}r Informatik},
  address =	{Dagstuhl, Germany},
  URL =		{https://drops.dagstuhl.de/entities/document/10.4230/LIPIcs.SAT.2023.27},
  URN =		{urn:nbn:de:0030-drops-184894},
  doi =		{10.4230/LIPIcs.SAT.2023.27}
}

@inproceedings{zulkoski2018learning,
  title={Learning-sensitive backdoors with restarts},
  author={Zulkoski, Edward and Martins, Ruben and Wintersteiger, Christoph M and Robere, Robert and Liang, Jia Hui and Czarnecki, Krzysztof and Ganesh, Vijay},
  booktitle={International Conference on Principles and Practice of Constraint Programming},
  pages={453--469},
  year={2018},
  organization={Springer}
}

@article{Buss1992GraphMultCounting,
title = {The graph of multiplication is equivalent to counting},
journal = {Information Processing Letters},
volume = {41},
number = {4},
pages = {199-201},
year = {1992},
issn = {0020-0190},
doi = {https://doi.org/10.1016/0020-0190(92)90180-4},
author = {Samuel R. Buss},
}

@book{Parberry1994CircuitComplexityNeuralNetworks,
  author    = {Ian Parberry},
  title     = {Circuit Complexity and Neural Networks},
  publisher = {The MIT Press},
  year      = {1994},
  address   = {Cambridge, MA},
  isbn      = {0262161486},
}

@inproceedings{mull2016hardness,
  title={On the hardness of {SAT} with community structure},
  author={Mull, Nathan and Fremont, Daniel J and Seshia, Sanjit A},
  booktitle={International Conference on Theory and Applications of Satisfiability Testing},
  pages={141--159},
  year={2016},
  organization={Springer}
}

@article{DBLP:journals/jcss/ImpagliazzoP01,
  author       = {Russell Impagliazzo and
                  Ramamohan Paturi},
  title        = {On the Complexity of k-{SAT}},
  journal      = {J. Comput. Syst. Sci.},
  volume       = {62},
  number       = {2},
  pages        = {367--375},
  year         = {2001},
  url          = {https://doi.org/10.1006/jcss.2000.1727},
  doi          = {10.1006/JCSS.2000.1727},
  bibsource    = {dblp computer science bibliography, https://dblp.org}
}

@article{DBLP:journals/jacm/AtseriasM20,
  author       = {Albert Atserias and
                  Moritz M{\"{u}}ller},
  title        = {Automating {Resolution} is {NP}-Hard},
  journal      = {J. {ACM}},
  volume       = {67},
  number       = {5},
  pages        = {31:1--31:17},
  year         = {2020},
  url          = {https://doi.org/10.1145/3409472},
  doi          = {10.1145/3409472},
  bibsource    = {dblp computer science bibliography, https://dblp.org}
}

@article{ansotegui_community_2019,
	title = {Community {Structure} in {Industrial} {SAT} {Instances}},
	volume = {66},
	issn = {1076-9757},
	url = {https://jair.org/index.php/jair/article/view/11741},
	doi = {10.1613/jair.1.11741},
	urldate = {2026-05-13},
	journal = {Journal of Artificial Intelligence Research},
	author = {Ansótegui, Carlos and Bonet, Maria Luisa and Giráldez-Cru, Jesús and Levy, Jordi and Simon, Laurent},
	month = oct,
	year = {2019},
	pages = {443--472},
}

@inproceedings{huang2004dpll,
author = {Huang, Jinbo and Darwiche, Adnan},
title = {Using DPLL for efficient OBDD construction},
year = {2004},
isbn = {354027829X},
publisher = {Springer-Verlag},
address = {Berlin, Heidelberg},
url = {https://doi.org/10.1007/11527695_13},
doi = {10.1007/11527695_13},
booktitle = {Proceedings of the 7th International Conference on Theory and Applications of Satisfiability Testing},
pages = {157–172},
numpages = {16},
location = {Vancouver, BC, Canada},
series = {SAT'04}
}

@inproceedings{imanishi2017upper,
  title={An upper bound for resolution size: Characterization of tractable sat instances},
  author={Imanishi, Kensuke},
  booktitle={International Workshop on Algorithms and Computation},
  pages={359--369},
  year={2017},
  organization={Springer}
}

@INPROCEEDINGS{wang2001cutwidth,
  author={Dong Wang and Clarke, E. and Yunshan Zhu and Kukula, J.},
  booktitle={Sixth IEEE International High-Level Design Validation and Test Workshop}, 
  title={Using cutwidth to improve symbolic simulation and Boolean satisfiability}, 
  year={2001},
  volume={},
  number={},
  pages={165-170},
  doi={10.1109/HLDVT.2001.972824}}

@article{Krajicek97,
  author  = {Jan Kraj{\'\i}{\v{c}}ek},
  title   = {Interpolation Theorems, Lower Bounds for Proof Systems, and
             Independence Results for Bounded Arithmetic},
  journal = {The Journal of Symbolic Logic},
  volume  = {62},
  number  = {2},
  pages   = {457--486},
  year    = {1997},
  doi     = {10.2307/2275541}
}

@article{Pudlak97,
  author  = {Pavel Pudl{\'a}k},
  title   = {Lower Bounds for Resolution and Cutting Plane Proofs and
             Monotone Computations},
  journal = {The Journal of Symbolic Logic},
  volume  = {62},
  number  = {3},
  pages   = {981--998},
  year    = {1997},
  doi     = {10.2307/2275583}
}
\newpage

\appendix

\section{Multiplicative Miter Instances as a special case of Theorem \ref{thm:proofdoor-main}}
\label{app:miterasthm1}

We now show how Theorem \ref{thm:proofdoor-main} is  more general result subsuming short proof of commutativity of integer multiplication instances. Consider the standard miter construction that checks $\mathbf{xy} \neq \mathbf{yx} $ where the outputs of the two multiplier circuits are compared bitwise. Let $e_0 , e_1 \cdots e_{2n-1}$ be the variables denoting the output difference variables, i.e. $e_i$ is true when output $i$ disagrees. The miter includes the clause $e_0 \vee e_1 \cdots e_{2n-1}$ asserting that some output bit differs.   Resolution refutations of this miter are based on decomposing the multiplier into \emph{critical strips}. Each strip covers a window of $\Delta = \log n$ consecutive columns of the multiplier array and contains all local constraints relevant to that window, including partial products, carry propagation, and the corresponding output bits. Under a partial assignment that sets exactly one output-difference bit  $e_k =1 $ and all earlier bits to $0$, the corresponding critical strip becomes unsatisfiable and admits a small ordered resolution refutation. 
  
We recast this argument in the proofdoor framework by defining a sequence of
(overlapping) chunks $A_0, A_1, \ldots, A_{2n-1}$, where each chunk $A_j$
is obtained from the critical strip $C_j$ (with the output (in)equality
constraints omitted) and consists of all multiplier constraints that affect
the outputs $e_{j-\Delta}, \ldots, e_j$. Intuitively, chunk $A_j$  captures the local constraints needed to rule out the possibility that the first output disagreement occurs within the window ending at position $j$. After refuting the prefix $A_0 \land \cdots A_j$ resolution derives the clause  $e_j$ along with  $e_{j+1} \vee \cdots e_{2n-1}$. Chunk $A_j$ shares the variables $e_{j - \Delta + 1} \cdots e_j$ with the subsequent formula. We therefore have the interpolant
$$
   I_j \;:=\; 
      (\neg e_{j-\Delta+1} \wedge \cdots \wedge \neg e_j)
      \;\wedge\;
      (e_{j-\Delta} \;\vee\; e_{j-\Delta+1} \;\vee\; \cdots \;\vee\; e_{2n-1}).
$$

Here $I_j$ is defined only on shared variables with the future chunks and also the number of clauses in $I_j$ is $O(\log n)$.

\section{DNF lower bound for \(\mathrm{EQ}_n\)}
\label{app:EQ_lower_bound}

We prove Lemma~\ref{EQ_lower_bound}, restated below.

\begingroup
\setcounter{theorem}{\numexpr\getrefnumber{EQ_lower_bound}-1\relax}
\begin{lemma}
\label{lem:EQ_lower_bound_restate}
Let \(f_{\EQ_n}:\{0,1\}^{2n}\to\{0,1\}\) be the equality function
$$f_{\EQ_n}(x,z) = \bigwedge_{i=0}^{n-1}(x_i=z_i).$$
Then any DNF computing \(f_{\EQ_n}\) has at least \(2^n\) terms.
\end{lemma}
\endgroup
\begin{proof}
Let $D = \bigvee_{j=1}^m T_j$
be a DNF computing $f_{\EQ_n}$, where each $T_j$ is a conjunction of literals over the variables $x_0,\dots,x_{n-1},z_0,\dots,z_{n-1}.$
Without loss of generality, assume no term is identically false and no term contains repeated literals. 

We claim that every term $T_j$ must mention every variable $x_i$ and every variable $z_i$. Fix a term $T=T_j$. Suppose that some term $T$ omits $x_i$ (the $z_i$ case is identical). Since $T$ is not identically false, there exists an assignment $(x,z)$ such that $T(x,z)=1$. Then $D(x,z)=1$, so $f_{\EQ_n}(x,z)=1$, which implies $x=z$. Now flip only the bit $x_i$, obtaining a new assignment $(x',z)$. Because $T$ does not mention $x_i$, we still have $T(x',z)=1$, hence $D(x',z)=1$. But now $x'\neq z$, so $f_{\EQ_n}(x',z)=0$, a contradiction.

Therefore each term contains one literal on each of the $2n$ variables, and hence is satisfied by at most one assignment in $\{0,1\}^{2n}$. Since $f_{\EQ_n}(x,z)=1$ on exactly the $2^n$ assignments with $x=z$, the DNF $D$ must have at least $2^n$ terms.
\end{proof}

\section{Implementation details for Floating-Point Addition Commutativity}

\label{app:floap}

\paragraph{Circuit implementation.}

\begin{enumerate}

\item { \bf Exponent comparison stage.}

We consider the addition of two positive normalized floating-point numbers
\(a = (E_a,M_a)\) and \(b = (E_b,M_b)\), where \(E\) denotes the exponent of $m$ bits and
\(M\) the mantissa of $n$ bits.

The exponent comparison stage computes bitwise equality and greater-than signals.  
Given \((E_a,E_b)\), for each bit position \(i\) the comparator produces
\[
\text{eq}_i = (E_{a,i} \leftrightarrow E_{b,i}),
\]
and defines a the signals:
\[
p_{m-1} = 1,
\qquad
p_i = (p_{i+1} \wedge \text{eq}_{i+1})
\quad \text{for } 0 \le i \le m-2.
\]

The greater-than and less-than bits are computed as
\[
\text{GT}_i = (p_i \wedge E_{a,i} \wedge \neg E_{b,i}),
\qquad
\text{LT}_i = (p_i \wedge \neg E_{a,i} \wedge E_{b,i}),
\]
and aggregated as
\[
\text{GT} = \bigvee_i \text{GT}_i,
\qquad
\text{LT} = \bigvee_i \text{LT}_i,
\qquad
\text{EQ} = \bigwedge_i \text{eq}_i.
\]

\begin{itemize}
\item {\bf CNF size:}  
Bitwise equalities add $m$ variables and $4m$ clauses.  
The $GT_i$ and $LT_i$ clauses add $m$ variables and $6m$ clauses.  
The $\text{GT}_i$ and $\text{LT}_i$ definitions add $2m$ variables and $12m$ clauses.
Aggregation adds $O(m)$ clauses.
\end{itemize}

\item {\bf Exponent selection.}

The larger and smaller exponents are selected via standard 2-input multiplexers:
\[
E_{large} = \mathrm{MUX}(\text{GT};\,E_a,E_b),\qquad
E_{small} = \mathrm{MUX}(\text{LT};\,E_a,E_b).
\]

\begin{itemize}
\item {\bf CNF size:} Adds $2m$ variables and $8m$ clauses.
\end{itemize}

\item {\bf Mantissa selection.}

Similarly, the corresponding mantissas are selected as
\[
M_{\text{large}} = \mathrm{MUX}(\text{GT};\,M_a,M_b),\qquad
M_{\text{small}} = \mathrm{MUX}(\text{GT};\,M_b,M_a).
\]

\begin{itemize}
\item {\bf CNF size:} Adds $2n$ variables and $8n$ clauses.
\end{itemize}

\item {\bf Exponent difference.}

The exponent difference is then computed using a conventional binary subtractor,
\[
\text{Diff} = E_{large} - E_{small}.
\]

\begin{itemize}
\item {\bf CNF size:} Adds $O(m)$ variables and $O(m)$ clauses.
\end{itemize}

\item {\bf Alignment stage.}

The smaller mantissa is aligned by a logical right shifter,
\[
M'_{\text{small}} = \mathrm{Shifter}(M_{\text{small}},\ \text{Diff}),
\]
which also produces the \emph{guard} (\(G\)), \emph{round} (\(R\)), and
\emph{sticky} (\(S\)) bits.  
Here \(G\) is the first bit below the stored LSB, \(R\) is the bit below \(G\),
and \(S\) is the OR of all remaining shifted-out bits.

\begin{itemize}
\item {\bf CNF size:} Barrel shifter with $m$ stages and $n$ MUXes per stage adds $nm$ variables and $4nm$ clauses (GRS logic contributes at most $O(nm)$ additional clauses).
\end{itemize}

\item {\bf Significand addition.}

After alignment, the circuit proceeds with significand addition. It produces
\[
\Sigma = M_{large } + M'_{small}.
\]

\begin{itemize}
\item {\bf CNF size:} Ripple adder adds $O(n)$ variables and $O(n)$ clauses.
\end{itemize}

\item {\bf Normalization.}

Let \(c\) denote the carry-out of this adder. Since both inputs are positive and normalized, the only possible normalization case after addition is overflow by one bit. This occurs precisely when the carry-out \(c=1\). Concretely, for the most significant bit we set
\[
\Sigma_{\text{norm}}[n-1] = \mathrm{MUX}(c;\, 1,\, \Sigma[n-1])
\]
and for all remaining bit positions $0 \le i \le n-2$ we wire
\[
\Sigma_{\text{norm}}[i] = \mathrm{MUX}(c;\, \Sigma[i+1],\, \Sigma[i]).
\]

The guard, round, and sticky bits are updated as:
\[
G_{\text{norm}} = \mathrm{MUX}(c; \Sigma[0] , G),
\]
\[
R_{\text{norm}} = \mathrm{MUX}(c; G ,R),
\]
\[
S_{\text{norm}} =
\mathrm{MUX}(c; R\vee S , S).
\]

\begin{itemize}
\item {\bf CNF size:} Uses $n$ MUXes, adding $n$ variables and $4n$ clauses (plus $O(1)$ for guard/round/sticky updates).
\end{itemize}

\item {\bf Exponent update.}

The exponent update is
\[
E_{\text{out}} = E_{\text{large}} + c.
\]

\begin{itemize}
\item {\bf CNF size:} $m$-bit incrementer adds $O(m)$ variables and $O(m)$ clauses.
\end{itemize}

\item {\bf Rounding.}

We assume round-to-nearest, ties-to-even.  
Let $\Sigma_{\text{norm}}[0]$ denote the least significant bit.  
The rounding increment is
\[
\text{inc} = R_{\text{norm}} \wedge
\bigl(\Sigma_{\text{norm}}[0] \vee S_{\text{norm}}\bigr).
\]

\begin{itemize}
\item {\bf CNF size:} Adds $O(1)$ variables and $O(1)$ clauses.
\end{itemize}

\item {\bf Rounded significand.}

The rounded significand is
\[
\Sigma_{\text{rnd}} = \Sigma_{\text{norm}} + \text{inc}.
\]

\begin{itemize}
\item {\bf CNF size:} $n$-bit incrementer adds $O(n)$ variables and $O(n)$ clauses.
\end{itemize}

\item {\bf Final shift and exponent update.}

Let $\kappa$ denote the carry-out of this incrementer.
A final conditional shift is applied only if $\kappa=1$:
\[
\Sigma_{\text{final}}[n-1] = \mathrm{MUX}(\kappa;\, 1,\, \Sigma_{\text{rnd}}[n-1]),
\]
\[
\Sigma_{\text{final}}[i] =
\mathrm{MUX}(\kappa;\, \Sigma_{\text{rnd}}[i+1],\, \Sigma_{\text{rnd}}[i])
\quad (0 \le i \le n-2).
\]

The exponent is updated as
\[
E_{\text{final}} = E_{\text{out}} + \kappa.
\]

\begin{itemize}
\item {\bf CNF size:} Uses $n$ MUXes, adding $n$ variables and $4n$ clauses.  
$m$-bit incrementer adds $O(m)$ variables and $O(m)$ clauses.
\end{itemize}

\end{enumerate}

{\bf Output and overall size.}

The output of the adder is given as $(\Sigma_{final} , E_{final})$.

The overall encoding introduces $O(nm + n + m)$ variables and clauses.

\paragraph{Proofdoor Decomposition and Small Proofdoor for Commutativity of Floating point addition}
We construct a miter between two identical circuits,
$$
L = \mathrm{add}(a,b),\qquad R = \mathrm{add}(b,a),
$$
and aim to show that all corresponding output bits agree. To encode inequivalence, we introduce error variables for each output bit.
For each exponent bit $i$ and mantissa bit $j$, define
$$
e^E_i \leftrightarrow (E_{\text{final},L}[i] \oplus E_{\text{final},R}[i]),
\qquad
e^M_j \leftrightarrow (\Sigma_{\text{final},L}[j] \oplus \Sigma_{\text{final},R}[j]).
$$

The miter formula asserts that at least one output bit differs, thus the SAT instance encodes
$$
\mathcal{F} \;=\; (L) \;\wedge\; (R) \;\wedge\; \left( \bigvee_i e^E_i \;\vee\; \bigvee_j e^M_j  \right)
$$

We show how we can decompose the circuit into small pathwidth chunks and the corresponding small interpolant conditions. We show the pathwidth analysis for a single side of the circuit, the small pathwidth for both Left and Right subparts immediately follows. The decomposition below partitions the miter into $O(m+n)$ chunks in total: 
stages $A_1$, $A_2$, and $A_6$ contribute $O(1)$ chunks each, 
stages $A_3$, $A_4$, and $A_7$ contribute $O(m)$ chunks each by slicing 
bit-by-bit, and stages $A_5$, $A_8$, and $A_9$ contribute $O(n)$ chunks 
each, giving $k = O(m+n)$ overall.

\begin{enumerate}
    \item \textbf{$A_1$}:

    We take the first chunk $A_1$ to be the full exponent comparator of both the left and right circuits. We first analyse the pathwidth of the comparator for a single circuit. The comparator consists of the equality gates $eq_i$, the chain variables $p_i$, the local comparison bits $GT_i,LT_i$, and the aggregate outputs $GT,LT,EQ$. We give a path decomposition with one bag $B_i$ per bit position $i$. Each bag $B_i$ contains the variables $E_{a,i},E_{b,i}$, $eq_i,eq_{i+1}$, $p_i,p_{i+1}$, $GT_i,LT_i$, and the aggregate outputs $GT,LT,EQ$. All clauses of the comparator CNF are covered by these bags. The clauses encoding $eq_i$ involve only $E_{a,i},E_{b,i},eq_i$ and lie in $B_i$. The clauses encoding $p_i$ involve $p_i,p_{i+1},eq_{i+1}$ and are contained in $B_i$. The clauses encoding $GT_i$ and $LT_i$ involve only $p_i,E_{a,i},E_{b,i},GT_i$ or $LT_i$, so they are also contained in $B_i$. For the aggregation $GT = \bigvee_i GT_i$, the clauses $(\neg GT_i \vee GT)$ lie in $B_i$, while the single clause $(\neg GT \vee GT_1 \vee \cdots \vee GT_m)$ contains only variables already present in the bags and can therefore be placed in all bags. The same argument applies to $LT$ and $EQ$. The bags are arranged in order $B_0,B_1,\dots,B_{m-1}$, and consecutive bags share the variables $p_{i+1}$ and $eq_{i+1}$, ensuring the running-intersection property. Since each bag contains only constantly many variables, the comparator chunk of a single circuit has constant clause--variable incidence pathwidth. Finally, this argument was for a single circuit. The comparator chunk of the miter contains both the left and right circuits. Since each copy has constant pathwidth, we obtain a path decomposition for the combined chunk by merging the corresponding bags of the two decompositions, which still yields constant width. The interpolant $I_1$ consists of the input-bit equalities together with the three comparator symmetries $GT_L \leftrightarrow LT_R$, $LT_L \leftrightarrow GT_R$, and $EQ_L \leftrightarrow EQ_R$. Thus $I_1$ has size $O(m+n)$ and the pathwidth of the chunk is $O(1)$.

    \item \textbf{$A_2$}

    We next consider the exponent and mantissa selection stage, which uses the comparator outputs to select the larger and smaller exponent and the corresponding mantissas, producing $E_{\text{large}},E_{\text{small}},M_{\text{large}},M_{\text{small}}$. This chunk is again a collection of bitwise multiplexers, so it has constant clause--variable incidence pathwidth: for each bit position, a bag contains the relevant input bits, the selector bit, and the corresponding output bit, and the constant-size Tseitin clauses of that MUX are covered within that bag. The interpolant for this chunk consists only of the equalities on the selected outputs, namely $E_{\text{large},L} \leftrightarrow E_{\text{large},R}$, $E_{\text{small},L} \leftrightarrow E_{\text{small},R}$, $M_{\text{large},L} \leftrightarrow M_{\text{large},R}$, and $M_{\text{small},L} \leftrightarrow M_{\text{small},R}$. Thus its size is $O(m+n)$. Moreover, each interpolant clause has constant support in the previous interpolant: the equality for one selected output bit follows from the constant-size MUX encoding together with only the corresponding selector equality and the corresponding input equalities from the previous interpolant. Hence every clause of this interpolant depends on only constantly many clauses of the previous interpolant.

    \item \textbf{$A_3$}

    We next consider the exponent-difference stage, but cut it into bit-slices rather than taking the full subtractor as one chunk. Thus for each bit position \(i\), we introduce a chunk \(A_{3,i}\) containing the local subtractor cell that computes the difference bit \(\mathit{Diff}_i\) and the outgoing borrow from the input bits \(E_{\mathrm{large},i},E_{\mathrm{small},i}\) and the incoming borrow. Each such chunk has constant clause--variable incidence pathwidth, since it consists of a single constant-size subtractor cell and its Tseitin encoding. The interpolant after chunk \(A_{3,i}\) is obtained by propagating the previous interpolant and adding the equalities produced at this slice. In particular, at slice \(i\) we add the clauses \(\mathit{Diff}_{i,L} \leftrightarrow \mathit{Diff}_{i,R}\) together with the equality of the outgoing borrow bit. Thus the interpolant after slice \(i\) contains the equalities for all difference bits computed so far as well as the current borrow equality. Each of these new clauses depends on only constantly many clauses of the previous interpolant, namely the equalities for \(E_{\mathrm{large},i}\), \(E_{\mathrm{small},i}\), and the incoming borrow, together with the constant-size CNF of the local subtractor cell. Hence the subtractor stage satisfies the small-support condition by construction, and after the final slice the interpolant contains the equalities for all bits of \(\mathit{Diff_L \leftrightarrow \mathit{Diff}_{R}}\).

    \item \textbf{$A_4$}

    We next consider the alignment stage, whose input is the selected smaller mantissa $M_{\mathrm{small}}$ and whose shift amount is determined by the exponent difference $\mathrm{Diff}$. Recall that the exponents have $m$ bits and the mantissas have $n$ bits. Accordingly, the alignment circuit is implemented as an $n$-bit logical right barrel shifter with $m$ stages, where stage $j$ conditionally shifts by $2^j$ positions according to the bit $\mathrm{Diff}_j$. Let $X^{(j)}=(x^{(j)}_{n-1},\dots,x^{(j)}_0)$ denote the intermediate word after stage $j$, with $X^{(0)}=M_{\mathrm{small}}$ and $X^{(m)}=M'_{\mathrm{small}}$. For each stage $j$ and bit position $i$, the output bit $x^{(j+1)}_i$ is defined by a 2-input multiplexer whose inputs are $x^{(j)}_i$ and either $x^{(j)}_{i+2^j}$ or $0$, depending on whether $i+2^j<n$. 
    
    We treat each stage as a separate chunk. The pathwidth argument is the same as in the previous MUX-based stage: for a single circuit, we take one bag per bit position $i$ containing the selector bit $\mathrm{Diff}_j$, the two inputs of the MUX computing $x^{(j+1)}_i$, and the output bit $x^{(j+1)}_i$. Since the Tseitin encoding of each MUX contributes only constantly many clauses over these variables, every clause of the stage is covered by the corresponding bag, and the bags form a path indexed by $i$. Thus each shifter stage has constant clause--variable incidence pathwidth, and the corresponding left-right miter chunk also has constant width by merging the two decompositions. 
    
    The interpolant after stage $j$ is obtained by propagating the previous interpolant equalities for the bits of $X^{(j)}$ together with the equality $\mathrm{Diff}_{j,L}\leftrightarrow \mathrm{Diff}_{j,R}$, and adding the equalities $x^{(j+1),L}_i \leftrightarrow x^{(j+1),R}_i$ for all $i$. Each such new clause depends on only constantly many clauses of the previous interpolant, namely the equality of the corresponding selector bit and the equalities of the corresponding MUX input wires, together with the constant-size CNF of the local MUX gate. After the final stage, the interpolant contains the equalities for all bits of the aligned mantissa $M'^{\,L}_{\mathrm{small},i} \leftrightarrow M'^{\,R}_{\mathrm{small},i}$ for all $i$ along with the exponent equalities
    
    \item \textbf{$A_5$}

    We next consider the significand addition stage, which computes $\Sigma = M_{\mathrm{large}} + M'_{\mathrm{small}}$. As in the exponent-difference stage $A_3$, we cut the ripple-carry adder into bit-slices rather than taking the full adder as one chunk. The pathwidth and constant-support arguments are identical to $A_3$, since each slice consists of a constant-size full-adder cell and its Tseitin encoding. The interpolant before slice $i$ contains the propagated equalities $E_{\mathrm{large},L} \leftrightarrow E_{\mathrm{large},R}$ together with $M_{\mathrm{large},i,L} \leftrightarrow M_{\mathrm{large},i,R}$, $M'_{\mathrm{small},i,L} \leftrightarrow M'_{\mathrm{small},i,R}$, and the equality of the incoming carry, from which we derive $\Sigma_{i,L} \leftrightarrow \Sigma_{i,R}$ and the equality of the outgoing carry. After the final slice, the interpolant contains $\Sigma_{i,L} \leftrightarrow \Sigma_{i,R}$ for all $i$ together with the propagated exponent equality.

    \item \textbf{$A_6$}

    We next consider the normalization stage. Let $c$ denote the carry-out of the significand adder. The normalized mantissa $\Sigma_{\mathrm{norm}}$ and updated exponent $E_{\mathrm{out}}$ are computed using multiplexers controlled by $c$: if $c=1$ the mantissa is shifted right by one position and the exponent is incremented, otherwise the values are propagated unchanged. This chunk consists only of bitwise MUX gates, so the pathwidth and constant-support arguments are identical to $A_2$. 
    
    The interpolant propagates the equalities $E_{\mathrm{large},L} \leftrightarrow E_{\mathrm{large},R}$ and $c_L \leftrightarrow c_R$ together with $\Sigma_{i,L} \leftrightarrow \Sigma_{i,R}$, from which we derive $\Sigma_{\mathrm{norm},i,L} \leftrightarrow \Sigma_{\mathrm{norm},i,R}$ and $E_{\mathrm{out},L} \leftrightarrow E_{\mathrm{out},R}$. Thus the interpolant after this stage contains the equalities for all bits of $\Sigma_{\mathrm{norm}}$ and the updated exponent.

    \item \textbf{$A_7$}

    We next consider the exponent update stage. The mantissa is propagated unchanged, so the interpolant simply carries the equalities $\Sigma_{\mathrm{norm},L} \leftrightarrow \Sigma_{\mathrm{norm},R}$. The exponent update computes $E_{\mathrm{out}} = E_{\mathrm{large}} + c$. As in stages $A_3$ and $A_5$, we cut this addition into bit-slices, and the pathwidth and constant-support arguments are identical to $A_3$. From the propagated equalities $E_{\mathrm{large},L} \leftrightarrow E_{\mathrm{large},R}$ and $c_L \leftrightarrow c_R$, we derive $E_{\mathrm{out},L} \leftrightarrow E_{\mathrm{out},R}$ bitwise. After the final slice, the interpolant contains the equalities for all bits of $E_{\mathrm{out}}$ together with the propagated mantissa equalities.

    \item \textbf{$A_8$}
    
    We next consider the rounding stage. The rounding increment is computed as $\mathrm{inc} = R_{\mathrm{norm}} \wedge (\Sigma_{\mathrm{norm}}[0] \vee S_{\mathrm{norm}})$, which introduces only $O(1)$ variables and clauses. The rounded significand is then computed as $\Sigma_{\mathrm{rnd}} = \Sigma_{\mathrm{norm}} + \mathrm{inc}$. As in stages $A_3$ and $A_5$, we cut this addition into bit-slices, so the pathwidth and constant-support arguments are identical to $A_3$. The interpolant propagates the equalities $\Sigma_{\mathrm{norm},L} \leftrightarrow \Sigma_{\mathrm{norm},R}$ together with $\mathrm{inc}_L \leftrightarrow \mathrm{inc}_R$, from which we derive $\Sigma_{\mathrm{rnd},L} \leftrightarrow \Sigma_{\mathrm{rnd},R}$ bitwise and the equality of the outgoing carry. After the final slice, the interpolant contains the equalities for all bits of $\Sigma_{\mathrm{rnd}}$.

    \item $A_9$
    
    The final chunk contains the conditional shift defining $\Sigma_{\mathrm{final}}$, the exponent update $E_{\mathrm{final}}=E_{\mathrm{out}}+\kappa$, and the disagreement clauses. We cut the mantissa computation into bit-slices. In the first slice (the MSB mantissa bit $i=n-1$), the chunk contains the MUX defining $\Sigma_{\mathrm{final}}[n-1]$, the clauses encoding $e_{n-1}^M \leftrightarrow (\Sigma_{\mathrm{final},L}[n-1] \oplus \Sigma_{\mathrm{final},R}[n-1])$, and the disagreement clause $\bigvee_{j=0}^{n-1} e_j^M \;\vee\; \bigvee_{t=0}^{m-1} e_t^E$. From the previous interpolant we have $\kappa_L \leftrightarrow \kappa_R$ and $\Sigma_{\mathrm{rnd},L}[n-1] \leftrightarrow \Sigma_{\mathrm{rnd},R}[n-1]$, which implies $\Sigma_{\mathrm{final},L}[n-1] \leftrightarrow \Sigma_{\mathrm{final},R}[n-1]$ and hence $\neg e_{n-1}^M$. Together with the disagreement clause this implies the shortened clause $\bigvee_{j=0}^{n-2} e_j^M \;\vee\; \bigvee_{t=0}^{m-1} e_t^E$. The interpolant after this slice therefore consists of the previous interpolant equalities for the remaining mantissa and exponent bits, $\kappa_L \leftrightarrow \kappa_R$, and the shortened clause above.

    For each remaining mantissa bit $0\le i\le n-2$, the chunk contains the MUX defining $\Sigma_{\mathrm{final}}[i]$ and the clauses encoding $e_i^M \leftrightarrow (\Sigma_{\mathrm{final},L}[i] \oplus \Sigma_{\mathrm{final},R}[i])$. From the previous interpolant we obtain $\Sigma_{\mathrm{final},L}[i] \leftrightarrow \Sigma_{\mathrm{final},R}[i]$, which shortens the clause by removing $e_i^M$. After the mantissa bits are exhausted, we slice the addition $E_{\mathrm{out}}+\kappa$ bit-by-bit from $0$ to $m-1$ exactly as in $A_3$, using the clauses $e_i^E \leftrightarrow (E_{\mathrm{final},L}[i] \oplus E_{\mathrm{final},R}[i])$ to eliminate the remaining literals.
    
    For the pathwidth argument, the MSB mantissa slice contains the MUX variables, the XOR-definition clauses for $e_{n-1}^M$, and the disagreement clause, which introduces all $e$ variables but only a constant blow-up in the pathwidth. Hence, this slice has a bounded clause–variable incidence pathwidth. All remaining mantissa slices contain strictly fewer variables because they do not include the disagreement clause and therefore have constant pathwidth. The exponent slices are identical to $A_3$ and hence have constant pathwidth. Moreover, each interpolant derivation depends only on the previous interpolant and the shortened clause, so the support size remains constant. The final interpolant is therefore the last shortened clause, which is a unit and has constant size.

\end{enumerate}

\end{document}